\documentclass[article]{elsarticle}
\usepackage{lineno,hyperref}
\modulolinenumbers[5]
\usepackage{amsmath}
\usepackage{amsfonts}
\usepackage{amssymb}
\usepackage{amsthm}
\usepackage{moreverb}
\usepackage{subfigure}
\usepackage{color}
\usepackage{epsfig}
\usepackage{epstopdf}
\usepackage{algorithm}
\usepackage{graphicx}
\usepackage{amssymb}
\usepackage{url} 
\usepackage{textcomp}
\usepackage{mathptmx}
\usepackage{epstopdf}
\usepackage{setspace}
\usepackage{eucal}
\usepackage{caption}
\newtheorem{theorem}{Theorem}[section]

\newtheorem{lemma}[theorem]{Lemma}
\usepackage[margin=1.3in]{geometry}
\hypersetup{
	colorlinks=true,
	citecolor=blue,
	filecolor=black,
	linkcolor=blue
}

\journal{Digital Signal Processing}

\begin{document}

\begin{frontmatter}

\title{Capacity and outage analysis of a dual-hop decode-and-forward relay-aided NOMA scheme}
%\author{Mohammed Belal Uddin$^{1}$}
%\ead{ahad.belal@kumoh.ac.kr}
\author{Md. Fazlul Kader$^{1}$\corref{cor1}}
\ead{f.kader@cu.ac.bd}
\cortext[cor1]{Corresponding author}
%\fntext[fn1]{Student}
\author{Mohammed Belal Uddin$^2$}
\ead{ahad.belal@kumoh.ac.kr}
\author{S. M. Riazul Islam$^3$}
\ead{riaz@sejong.ac.kr}
\author{Soo Young Shin$^2$}
\ead{wdragon@kumoh.ac.kr}

%\fntext[fn2]{Lecturer}

%\author{Md Fazlul Kader\corref{cor1}}
%\ead{f.kader@cu.ac.bd}
%\author{Soo Young Shin}
%\ead{wdragon@kumoh.ac.kr}
\address{$^1$Department of Electrical and Electronic Engineering, 
	University of Chittagong, Chittagong-4331, Bangladesh.}
\address{$^2$Department of IT Convergence Engineering, Kumoh National Institute of Technology, Gumi 39177, South Korea}
\address{$^3$Department of Computer Science and Engineering,
	Sejong University, Seoul 05006, South Korea\\}




%\cortext[]{E-mail address: f.kader@cu.ac.bd}
%\corraddr{Soo Young Shin, WENS Lab., Kumoh National Institute of Technology, Republic of Korea. email: wdragon@kumoh.ac.kr}


\begin{abstract}
Non-orthogonal multiple access (NOMA) is regarded as a candidate radio access technique for the next generation wireless networks because of its manifold spectral gains. A two-phase cooperative relaying strategy (CRS) is proposed in this paper by exploiting the concept of both downlink and uplink NOMA (termed as DU-CNOMA). In the proposed protocol, a transmitter considered as a source transmits a NOMA composite signal consisting of two symbols to the destination and relay during the first phase, following the principle of downlink NOMA. In the second phase, the relay forwards the symbol decoded by successive interference cancellation to the destination, whereas the source transmits a new symbol to the destination in parallel with the relay, following the principle of uplink NOMA. The ergodic sum capacity, outage probability, outage sum capacity, and energy efficiency are investigated comprehensively along with analytical derivations, under both perfect and imperfect successive interference cancellation. To inquire more insight into the system outage performance, diversity order for each symbol in the proposed DU-NOMA is also demonstrated. The performance improvement of the proposed DU-CNOMA over the conventional CRS using NOMA is proved through analysis and computer simulation. Furthermore, the correctness of the author's analysis is proved through a strong agreement between simulation and analytical results.
\end{abstract}
\begin{keyword}
	Cooperative relaying \sep Downlink\sep Ergodic capacity \sep Non-orthogonal multiple access \sep  Successive interference cancellation \sep Uplink.
\end{keyword}
\end{frontmatter}

\linenumbers
%\doublespacing
\section{Introduction}
To deal with the high data rate requirements of the next generation wireless networks, integration of multiple technologies is anticipated~\cite{1, 2}. Cooperative relaying strategy (CRS) is an important technology for wireless networks to improve system capacity, combat fading, and extend service coverage~\cite{3, 4, 5, 6}. In addition, non-orthogonal multiple access (NOMA) has garnered substantial attention from the industry and academia to meet with the large data rate requirements of 5G and beyond~ \cite{7, 8, 9, 10}. In this paper, cooperative diversity and NOMA are integrated, which can be a promising approach to meet the capacity demands for future wireless networks~\cite{11, 12, 13}.

A lot of different varieties of NOMA can be found in the literature~\cite{14}. Among them, cooperative NOMA (C-NOMA) is one of the most dynamic areas of research \cite{11, 12, 13, 14, 15, 16}. Based on the direction of data transmission C-NOMA can be further classified as uplink C-NOMA~\cite{17, 18} and downlink C-NOMA~\cite{19, 20, 21}.
In \cite{19}, a cooperative relaying scheme (CRS) using NOMA (termed as CRS-NOMA) was proposed to improve the spectral efficiency over independent Rayleigh fading channels, where a source (S) transmits a superposition coded composite signal to the relay (R) and the destination (D), during the first time slot. Then, R decodes the symbol to be relayed by performing successive interference cancellation (SIC), whereas D decodes own symbol considering other signal as noise. In the subsequent time slot, R retransmits the decoded symbol with full power to D. The outcome of CRS-NOMA~\cite{19} demonstrates that CRS-NOMA can achieve better sum capacity than traditional decode-and-forward~\cite{3} for a high signal-to-noise
ratio (SNR) $\rho$, but shows worse performance for a low $\rho$. In~\cite{20}, the performance of CRS-NOMA~\cite{19} was investigated over Rician fading channels. A novel detection scheme for CRS-NOMA~\cite{19} was proposed in~\cite{21} (termed as CRS-NOMA-ND), where D uses maximal-ratio combining (MRC) and another SIC to jointly decode transmitted symbols by source. It was shown that CRS-NOMA-ND~\cite{21} outperforms CRS-NOMA~\cite{19}, particularly when the link between S and R is better than the R to D link. Note that only achievable average rate was analyzed in~\cite{19, 20, 21}.  In~\cite{22, 23}, the authors considered both uplink and downlink NOMA under non-cooperative scenario, whereas~\cite{24, 25} exploited the concept of both uplink and downlink NOMA under cooperative scenario. In~\cite{24}, a cooperative relay sharing network was proposed, where multiple sources can communicate with their corresponding destination simultaneously through a common relay. A C-NOMA scheme considering both downlink and uplink transmission systems, was proposed in ~\cite{25}, where a strong user works as a cooperative relay for the weak user.

Unlike the existing works, in this paper, a CRS-NOMA scheme using the concept of downlink and uplink NOMA (termed as DU-CNOMA) is proposed. In the proposed DU-CNOMA, S transmits a superposition coded composite signal consisting of two symbols $s_1$ and $s_2$ to D and R, according to the principle of downlink NOMA as in~\cite{19, 20, 21}, during the first time slot. However, in the subsequent time slot, unlike~\cite{19, 20, 21}, S transmits a new symbol $s_3$ and R transmits decoded symbol $s_2$ to D simultaneously, according to the principle of uplink NOMA \cite{17, 22, 23}. Furthermore, unlike~\cite{19, 20, 21}, where only perfect SIC is considered, we consider both perfect and imperfect SIC by taking into account a more realistic scenario. Principal contributions of this paper are outlined as follows: 
	\begin{enumerate}
		\item A dual-hop CRS by taking into account NOMA is proposed and investigated over independent Rayleigh fading channels.
		\item The closed-form expressions of the ergodic sum capacity (ESC), outage probability (OP), and outage sum capacity (OSC) of DU-CNOMA are derived under both perfect and imperfect SIC. The analytical results are validated by Monte Carlo simulation. To look into more insight of the proposed DU-NOMA system, diversity order (DO) for each symbol is also investigated.
		\item The energy efficiency (EE) of DU-NOMA is computed. Moreover, for the purpose of comparison, EE of CRS-NOMA~\cite{19} and CSR-NOMA-ND~\cite{21} is also computed.
		\item The performance improvement of the proposed DU-CNOMA over CRS-NOMA~\cite{19}, and CSR-NOMA-ND~\cite{21} is manifested through analysis and simulation. The outcomes of this paper demonstrate that the proposed DU-NOMA is more spectral and energy efficient as compared to CRS-NOMA~\cite{19} and CSR-NOMA-ND~\cite{21}.
\end{enumerate}
The rest of this paper is organized as follows. The system model with detailed description of the proposed protocol is provided in Section 2. The channel model is also demonstrated in this section. The closed-form expressions of the ESC, OP, and OSC are presented in Section 3, 4, and 5, respectively. In Section 6, EE is investigated. The numerical results that are validated by Monte Carlo simulation are provided in Section 7, and finally, the conclusion is drawn along with future recommendations in Section 8.
\begin{figure}[t]
	\centering
	\includegraphics[width=6in,height=6in,keepaspectratio]{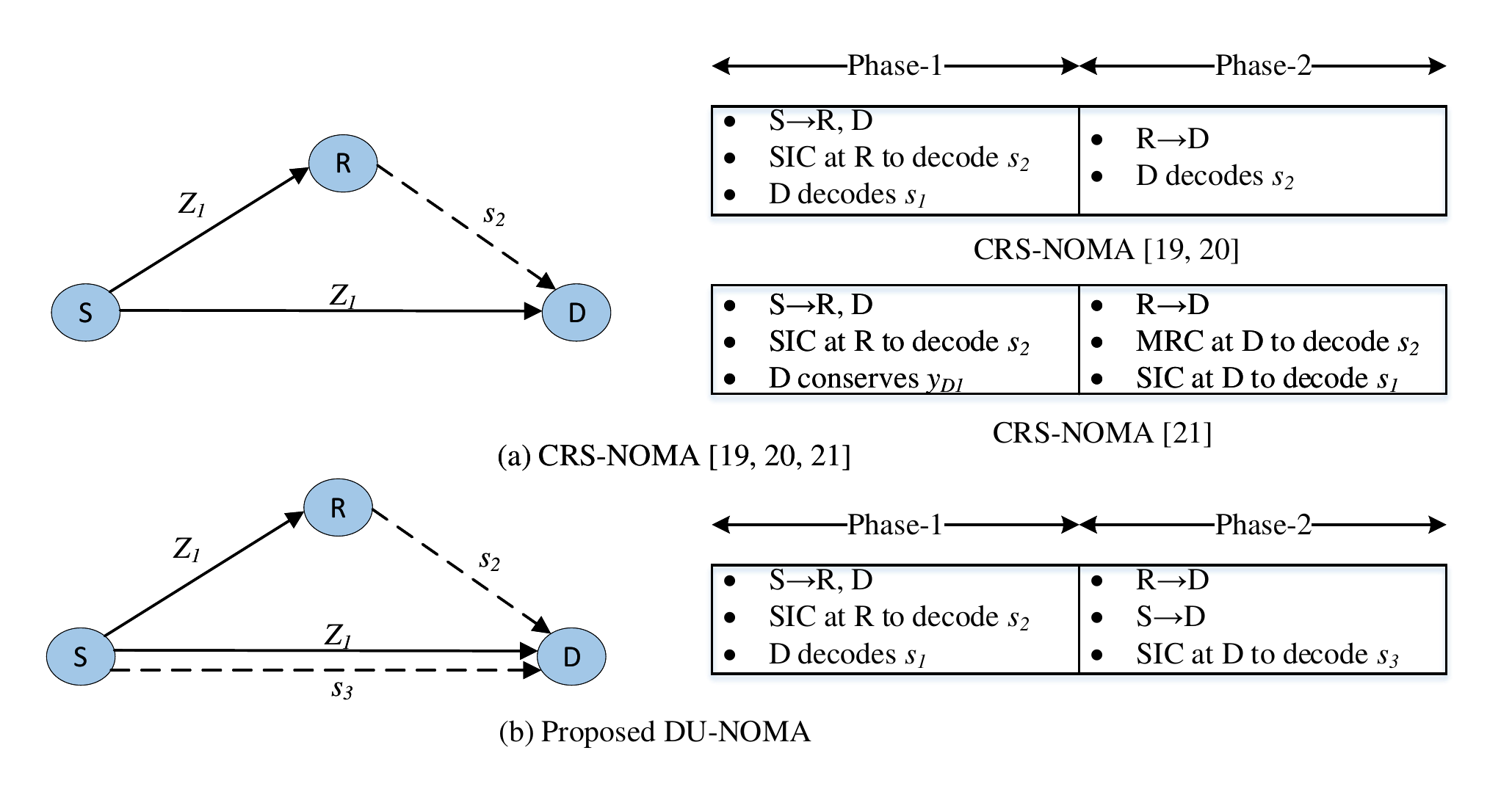}
\caption{System model: (a) CRS-NOMA~\cite{19, 20, 21} and (b) Proposed DU-CNOMA.}
	\label{Fig1}
\end{figure}
\begin{figure}[t]
	\centering
	\includegraphics[width=4in,height=10in,keepaspectratio]{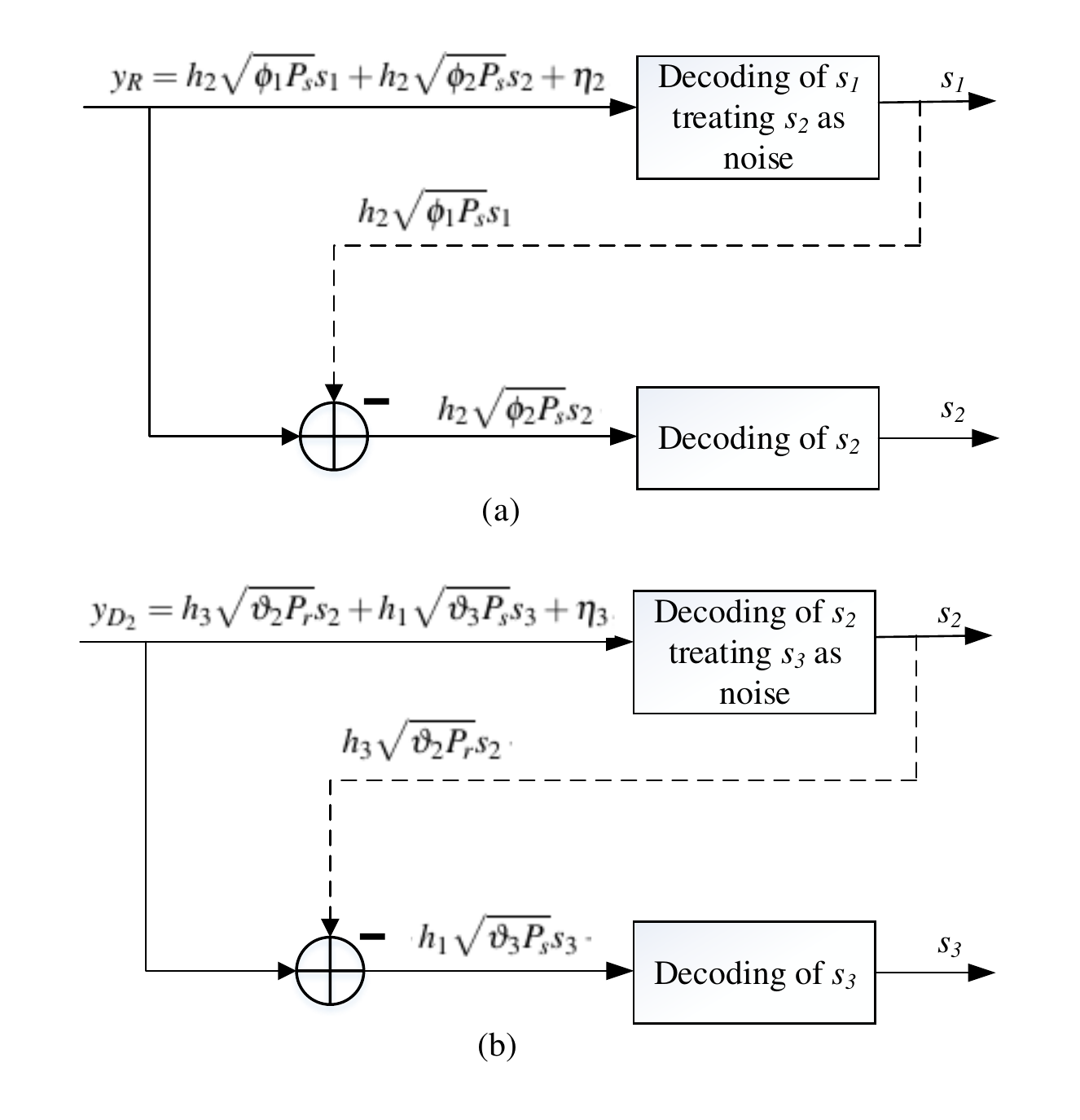}
	\caption{SIC process (a) at R following the principle of downlink NOMA  and (b) at D following the principle of uplink NOMA, in DU-NOMA.}
	\label{Fig1a}
\end{figure}

\section{Network architecture and protocol descriptions}
A half-duplex cooperative relaying protocol exploiting the concept of both downlink and uplink NOMA is proposed. The system architecture consists of a source (S), a DF relay (R), and a destination (D), as drawn in Fig. \ref{Fig1}. Fig. \ref{Fig1} (a) shows the system architecture considered in~\cite{19, 20, 21}, whereas Fig. \ref{Fig1} (b) shows the system architecture of the proposed DU-NOMA, with related illustrations.
%\footnote{Note that Fig. \ref{Fig1} is an equivalent to a network with a BS (S), a %NOMA-strong user (R), and a NOMA-weak user (D), where the NOMA-strong user acts as a %cooperative relay to the NOMA-weak user.}.
All the links (i.e., S-to-R, S-to-D, and R-to-D) are considered available and subjected to independent Rayleigh fading. Channel coefficient with zero mean and variance $\lambda_{i}=d_{i}^{-\nu}$ is represented by $h_{i}\sim{CN}(0,\lambda_{i})$, where $d$ is the distance, $\nu$ is the path loss exponent, and $i\in\{1,2,3\}$. Parameters $h_1$, $h_2$, and $h_3$ refer to the respective complex channel coefficient of S-to-D, S-to-R, and R-to-D links. Without loss of generality, it is assumed that  $\lambda_{1}<\lambda_{2}$ and $\lambda_{1}<\lambda_{3}$, under statistical channel state information~\cite{25}. The data transmission in the proposed protocol is performed by two cooperative phases as follows. 
\subsection{Phase-1 ($t_1$)}
At the first phase of the transmission, the source S transmits a composite NOMA signal $Z_1=\sqrt{\phi_1P_s}s_1+\sqrt{\phi_2P_s}s_2$ consisting of two symbols $s_1$ and $s_2$ to D and R simultaneously as per law of downlink NOMA. The symbols $s_1$ and $s_2$ are corresponded to D and R, respectively. The total transmit power of S, the power allocation factor with $s_1$, and the power allocation factor with $s_2$ are respectively denoted by $P_s$, $\phi_1$, and $\phi_2$ wherein $\phi_1>\phi_2$ and $\phi_1+\phi_2=1$. Thus, the received signals at D and R are respectively written as
\begin{align}
y_{D_1}=h_1\sqrt{\phi_1P_s}s_1+h_1\sqrt{\phi_2P_s}s_2+\eta_1, \label{eq1.1}\\
y_R=h_2\sqrt{\phi_1P_s}s_1+h_2\sqrt{\phi_2P_s}s_2+\eta_2, \label{eq1.2}
\end{align}
where $\eta_1$ and $\eta_2$ are the complex additive white Gaussian noise (AWGN) at D and R, respectively with zero mean and variance $\sigma^2$.

Upon receiving the signal, firstly, R extracts $s_1$ by treating $s_2$ as noise. Then, it performs SIC to cancel out the extracted information from the received signal and thus it extracts $s_2$. For a clear understanding, the SIC process at R following the principle of downlink NOMA, is pictorially represented in Fig. 2 (a). The received signal-to-interference plus noise ratios (SINRs) at R for symbols $s_1$ and $s_2$ are respectively represented by 
\begin{align}
\gamma^{t_1}_{s_1\rightarrow s_2} =& \frac{\phi_1P_s |h_2|^2}{\phi_2P_s|h_2|^2+\sigma^2}=\frac{\phi_1\rho |h_2|^2}{\phi_2\rho|h_2|^2+1}, \label{eq1} \\
%\gamma^{t_1}_{s_{2}} =& \frac{\phi_2\rho |h_2|^2}{\phi_1\rho|\bar{f_2}|^2+1}, 
\gamma^{t_1}_{s_{2}} =& \frac{\phi_2\rho |h_2|^2}{\phi_1\rho |\bar{h}_{2}|^2+1},
\label{eq2}
\end{align}
where $\bar{h}_{2}\sim CN(0,\kappa_1\lambda_{2})$, $\rho\triangleq \frac{ P_s}{\sigma^2}$ is the transmit SNR of S and $\sigma^2$ is the noise variance. It is noted that several potential implementation issues (i.e, error propagation and complexity scaling) with the use of SIC may lead to errors in decoding even under the perfect channel estimation assumption.  As a result the interference may not be removed completely and there exists residual interference~\cite{17, 24, 26}. Hence, in the proposed DU-NOMA, the parameter $\kappa_1$ ($0<\kappa_1\leq1$) represents the level of residual interference at R because of imperfect SIC. As a particular case, $\kappa_1=0$ represents perfect SIC i.e., the interference is canceled completely and no residual interference exists. Contrarily,  $\kappa_1=1$ represents that SIC is not performed at R at all, which means that R has to consider $s_1$ as interference to decode $s_2$.

On the other hand, D decodes $s_1$ by recking of $s_2$ as noise. So, the received SINR regarding symbol $s_{1}$ at D is obtained as 
\begin{align}\label{eq3}
\gamma^{t_1}_{s_{1}} ={}& \frac{\phi_1\rho |h_{1}|^2}{\phi_2\rho|h_{1}|^2+1}.
\end{align}
\subsection{Phase-2 ($t_2$)}
During the second phase, according to the law of uplink NOMA, R retransmits the decoded symbol $s_2$ and S transmits a new symbol $s_3$ to D at the same instant of time. The respective assigned powers with $s_2$ and $s_3$ are $\sqrt{P_r\vartheta_2}$ and $\sqrt{P_s\vartheta_3}$, respectively. The total transmit power of R, the power allocation factor with $s_2$, and the power allocation factor with $s_3$ are respectively denoted by $P_r$, $\vartheta_2$, and $\vartheta_3$ wherein $\vartheta_2>\vartheta_3$. The received signal at D is therefore given by
\begin{align}
y_{D_2}=h_3\sqrt{\vartheta_2P_r}s_2+h_1\sqrt{\vartheta_3P_s}s_3+\eta_3, \label{eq1.3}
\end{align}
where $\eta_3$ is the complex AWGN at D during $t_2$ with zero mean and variance $\sigma^2$. As the information related to $s_2$ is dominant over $s_3$ at the destination, D first decodes $s_2$ by considering $s_3$ as noise. After then, by applying SIC procedure, it subtracts the decoded information to get $s_3$. For a clear understanding, the SIC process at D following the principle of uplink NOMA is pictorially represented in Fig. 2 (b). The received SINRs concerning $s_2$ and $s_3$ at D are respectively given by 
\begin{align}
\gamma^{t_2}_{s_2} =& \frac{\vartheta_2P_r |h_3|^2}{\vartheta_3P_s|h_1|^2+\sigma^2}=\frac{\vartheta_2\rho |h_3|^2}{\vartheta_3\rho|h_1|^2+1}, \label{eq4} \\
\gamma^{t_2}_{s_{3}} =& \frac{\vartheta_3\rho |h_1|^2}{\vartheta_2\rho |\bar{h}_{3}|^2+1}, \label{eq5}
\end{align}
where $\bar{h}_{3}\sim CN(0,\kappa_2\lambda_{3})$, $\rho\triangleq \frac{ P_r}{\sigma^2}$ is the transmit SNR of R, and $\kappa_2$ represents the level of residual interference at D. Note that $\kappa_2$ shows similar behavior like $\kappa_1$.
\subsection{Sum capacity}
The end-to-end data rate of a multi-hop cooperative network is determined by the weakest link. So, the achievable rate related to $s_1$ is depicted by 
\begin{align}\label{eq6}
C_{1} &=\frac{1}{2} \log_2\left(1+\min\left(\gamma^{t_1}_{s_1\rightarrow s_2},\gamma^{t_1}_{s_{1}}\right)\right),
\end{align}
The achievable rate associated with $s_2$ is dependent on (\ref{eq2}) and (\ref{eq4}), which can be denoted by 
\begin{align}\label{eq7}
C_{2} &=\frac{1}{2} \log_2\left(1+\min\left(\gamma^{t_1}_{s_{2}},\gamma^{t_2}_{s_2}\right)\right),
\end{align}
By using (\ref{eq5}), the achievable rate related to $s_3$ is given by 
\begin{align}\label{eq8}
C_{3} &=\frac{1}{2} \log_2\left(1+\gamma^{t_2}_{s_{3}}\right),
\end{align}
Therefore, the sum capacity of the proposed DU-CNOMA system can be calculated by summing up (\ref{eq6}), (\ref{eq7}), and (\ref{eq8}) as follows
\begin{align}\label{eq9}
C_{\textup{sum}}^{\textup{prop}} &=C_{1}+C_{2}+C_{3}.
\end{align}
%Conversely, based on the fact that the end-to-end capacity of a two-hop cooperative link is the minimum one of the two hops, and using (\ref{eq2}), (\ref{eq3}) and (\ref{eq5}), the capacity of CEU is obtained as 
%\begin{align}\label{eq7}
%R_{e} &= 
%\begin{cases} 
%\min(R^{t_1}_{y_1, 1}, R^{t_2}_{y_1}), \quad\textup{C-1} \\
%\min(R^{t_1}_{y_1, 2}, R^{t_2}_{y_1}), \quad\textup{C-2}\nonumber\\
%\end{cases}
%\\
%& = 
%\begin{cases} 
%\frac{1}{2} \log_2\bigl(1+\min\bigl(\gamma^{t_1}_{y_1, 1}, %\gamma^{t_2}_{y_1}\bigr)\bigr), \quad\textup{C-1}\\
%\frac{1}{2} \log_2\biggl(1+\min\bigl(\gamma^{t_1}_{y_1, 2}, %\gamma^{t_2}_{y_1}\bigr)\bigr), \quad\textup{C-2}
%\end{cases}
%\end{align}
\section{Capacity analysis}
The closed-form ESC expression of the proposed DU-CNOMA is derived over independent Rayleigh fading channel, in this section.
\subsection{Ergodic capacity related to $s_1$}
The achievable rate of (\ref{eq6}), can be simplified as \cite[eq. (8)]{19}
\begin{align}\label{eq10}
C_{1} &=\frac{1}{2} \log_2\left(1+\min\{|h_1|^2, |h_2|^2\}\rho\right)\nonumber\\
&-\frac{1}{2} \log_2\left(1+\min\{|h_1|^2, |h_2|^2\}\rho\phi_2\right),
\end{align}
Let $W\triangleq\text{min}\left(|h_1|^2,|h_2|^2\right)$. Applying PDF $f_{|h_\iota|^2}\left(w\right)=\left(1/\lambda_{\iota}\right)e^{-w/\lambda_{\iota}}$ for $\iota\in\{1,2\}$, the CDF of $W$ is derived as $F_W\left(w\right)=1-e^{-w\left(\frac{1}{\lambda_{1}}+\frac{1}{\lambda_{2}}\right)}$. Then, the probability density function of $W$ is derived by taking the derivative of $F_W\left(w\right)$ as 
\begin{align}\label{eq11}
f_W\left(w\right)=\left(\frac{1}{\lambda_{1}}+\frac{1}{\lambda_{2}}\right)e^{-w\left(\frac{1}{\lambda_{1}}+\frac{1}{\lambda_{2}}\right)}
\end{align}
Now, using (\ref{eq10}) and (\ref{eq11}), the EC associated with $s_1$ can be obtained as
\begin{align}\label{eq12}
\bar{C}^{\textup{ex}}_{1} &=\text{E}\{C_{1}\}\nonumber\\ 
&=\frac{1}{2}\displaystyle \int^{\infty}_{0}\left\{\log_2\left(1+\rho w\right)-\log_2\left(1+\rho w\phi_2\right)\right\}f_{W}(w)\,dw
\end{align}
Using $\log_2 (x)=\frac{\ln (x)}{\ln 2}$, (\ref{eq12}) can be written as
\begin{align}\label{eq13}
\bar{C}^{\textup{ex}}_{1} &=\frac{1}{2 \ln 2}\displaystyle \int^{\infty}_{0}\left\{\ln\left(1+\rho w\right)-\ln\left(1+\rho w\phi_2\right)\right\}f_{W}(w)\,dw
\end{align}
By applying $\displaystyle \int^{\infty}_{0}e^{-mw}\ln\left(1+nw\right)\,dw=-\frac{1}{m}e^{m/n}\text{Ei}\left(-m/n\right)$~\cite[eq. (4.337.2)]{27}, 
\begin{align}\label{eq14}
\bar{C}^{\textup{ex}}_{1}&=-\frac{1}{2\ln 2}e^{\frac{1}{\rho}\left(\frac{1}{\lambda_{1}}+\frac{1}{\lambda_{2}}\right)}\text{Ei}\left(-\frac{1}{\rho}\left(\frac{1}{\lambda_{1}}+\frac{1}{\lambda_{2}}\right)\right)\nonumber\\
&+\frac{1}{2 \ln 2}e^{\frac{1}{\phi_2\rho}\left(\frac{1}{\lambda_{1}}+\frac{1}{\lambda_{2}}\right)}\text{Ei}\left(-\frac{1}{\phi_2\rho}\left(\frac{1}{\lambda_{1}}+\frac{1}{\lambda_{2}}\right)\right)
\end{align} 
where $\text{E}\{\cdot\}$ and $\text{Ei}\{\cdot\}$ denote the expectation operator and exponential integral function, respectively~\cite{27}.
\subsection{Ergodic capacity related to $s_2$}
Let $U\triangleq\gamma^{t_1}_{s_2}$, $V\triangleq\gamma^{t_2}_{s_2}$, and $Z\triangleq\text{min}\left(U,V\right)$. The CDF of $U$ and $V$ can be respectively written as~\cite[eq. (7)]{24}
\begin{align}
F_U\left(u\right) &=1-\frac{\phi_2\lambda_{2}}{\phi_2\lambda_{2}+\phi_1\kappa_1\lambda_{2}v}e^{-\frac{v}{\phi_2\rho\lambda_{2}}}\nonumber\\
&=1-\frac{p}{p+u}e^{-\frac{u}{\phi_2\rho\lambda_2}},
\label{eq15} \\
F_V\left(v\right) &=1-\frac{\vartheta_2\lambda_{3}}{\vartheta_2\lambda_{3}+\vartheta_3\lambda_{1}v}e^{-\frac{v}{\vartheta_2\rho\lambda_{3}}}\nonumber\\
&=1-\frac{g}{g+v}e^{-\frac{v}{\vartheta_2\rho\lambda_{3}}},
\label{eq16}
\end{align}
where $p=\phi_2 \lambda_{2}/\phi_1\kappa_1\lambda_{2}$ and $g=\vartheta_2\lambda_{3}/\vartheta_3\lambda_{1}$. Using (\ref{eq15}) and (\ref{eq16}), the CDF of $Z$ can be obtained as
\begin{align}\label{eq17}
F_Z\left(z\right) 
&=1-\frac{pg}{(p+z)(g+z)}e^{-qz},
\end{align}
where $q=\frac{1}{\phi_2\rho\lambda_{2}}+\frac{1}{\vartheta_2\rho\lambda_{3}}$. So, the exact EC related to $s_2$ is derived as
\begin{align}\label{eq18}
\bar{C}^{\textup{ex}}_{2} ={}&\text{E}\{C_{2}\} =\frac{1}{2}\displaystyle \int^{\infty}_{0}\log_2\left(1+z\right)f_{Z}(z)\,dz.
\end{align}
Applying $\textstyle \int^{\infty}_{0}\log_2\left(1+z\right)f_{Z}(z)\,dz=\frac{1}{\ln2}\textstyle \int^{\infty}_{0}\frac{1-F_{Z}(z)}{1+z}\,dz$, (\ref{eq18}) can be represented as
\begin{align}\label{eq19}
\bar{C}^{\textup{ex}}_{2}
&=\frac{1}{2\ln 2}
\displaystyle \int^{\infty}_{0}\frac{pg}{\left(1+z\right)(p+z)(g+z)}e^{-qz}dz\nonumber \\
=&\frac{p\log_2 e}{2(p-1)}\int_{0}^{\infty}\biggr[(1+z)^{-1}\frac{g}{g+z}-(p+z)^{-1}\frac{g}{g+z}  \biggr]e^{-qz}dz\nonumber \\
=&\frac{p\log_2 e}{2(p-1)}\biggr[ \frac{g}{g-1}\biggl\{-e^q\operatorname{Ei}(-q)+e^{gq}\operatorname{Ei}(-gq)\biggr\}\nonumber \\
&-\frac{g}{g-p}\biggl\{-e^{pq}\operatorname{Ei}(-pq)+e^{gq}\operatorname{Ei}(-gq)\biggr\} \biggr].
\end{align}
%Note that under perfect SIC (i.e., $\kappa_1$=0), (16) is not reasonable. Therefore, 
%the exact EC of $s_2$ under perfect SIC is derived as follows.
Note that (\ref{eq19}) is derived by considering imperfect SIC (i.e., 0$<\kappa_1\leq$1). Therefore, the exact EC of $s_2$ under perfect SIC is derived as follows.

With perfect SIC, $Z\triangleq\text{min}\left(U,V\right)$ can be written as $Z\triangleq\text{min}\left(\phi_2\rho |h_2|^2,V\right)$. The CDF of $Z$ is therefore given by 
\begin{align}\label{eq20}
F_Z\left(z\right) 
&=1-\frac{g}{g+z}e^{-qz},
\end{align}
The exact EC related to $s_2$ under perfect SIC, is derived as
\begin{align}\label{eq21}
\bar{C}^{\textup{ex}}_{2, \textup{p}} &=\frac{1}{2}\displaystyle \int^{\infty}_{0}\log_2\left(1+z\right)f_{Z}(z)\,dz \nonumber\\
&=\frac{\log_2e}{2}
\displaystyle \int^{\infty}_{0}\frac{g}{\left(1+z\right)\left(g+z\right)}e^{-qz}dz\nonumber\\
&=\frac{g\log_2e}{2\left(g-1\right)}
\displaystyle \int^{\infty}_{0}\left(\frac{1}{1+z}-\frac{1}{g+z}\right)e^{-qz}dz\nonumber\\
&=\frac{g\log_2e}{2\left(g-1\right)}
\left\{-e^{q}\text{Ei}\left(-q\right)+e^{gq}\text{Ei}\left(-gq\right)\right\}.
\end{align}
\subsection{Ergodic capacity related to $s_3$}
Let $Y\triangleq\gamma^{t_2}_{s_3}$. So, the CDF of $Y$ is derived as
\begin{align}\label{eq22}
F_Y\left(y\right) &=1-\frac{\vartheta_3 \lambda_{1}}{\vartheta_3 \lambda_{1}+\vartheta_2\kappa_2\lambda_{3}y}e^{-\frac{y}{\vartheta_3\rho\lambda_{1}}}.
\end{align}
So, the exact EC associated with $s_3$ is obtained as 
\begin{align}\label{eq23}
&\bar{C}^{\textup{ex}}_{3} =\text{E}\{C_{3}\} =\frac{1}{2}\displaystyle \int^{\infty}_{0}\log_2\left(1+y\right)f_{Y}(y)\,dy\nonumber\\
&=\frac{\log_2 e}{2}\displaystyle \int^{\infty}_{0} (1+y)^{-1}\frac{\vartheta_3 \lambda_{1}}{\vartheta_3 \lambda_{1}+\vartheta_2\kappa_2\lambda_{3}y}e^{-\frac{y}{\vartheta_3\rho\lambda_{1}}}\,dy \nonumber\\
=&\frac{\log_2e}{2}\biggl[ \frac{\vartheta_3 \lambda_{1}}{\vartheta_3 \lambda_{1}-\vartheta_2\kappa_2\lambda_{3}}  \nonumber \\
&\times\biggl\{ \int_{0}^{\infty} (1+y)^{-1}-\int_{0}^{\infty}\frac{\vartheta_2\kappa_2\lambda_{3}}{\vartheta_3 \lambda_{1}+\vartheta_2\kappa_2\lambda_{3}v}\biggr\}e^{-\frac{y}{\vartheta_3\rho\lambda_{1}}}dy \biggr] \nonumber\\
=&\frac{\log_2e}{2}\biggl[\frac{\vartheta_3 \lambda_{1}}{\vartheta_3 \lambda_{1}-\vartheta_2\kappa_2\lambda_{3}} \times \nonumber \\
&\biggl\{ -e^{\frac{1}{\vartheta_3\rho\lambda_{1}}} \operatorname{Ei} (-\frac{1}{\vartheta_3\rho\lambda_{1}}) +e^{\frac{1}{\vartheta_2\kappa_2\rho\lambda_{3}}} \operatorname{Ei} (-\frac{1}{\vartheta_2\kappa_2\rho\lambda_{3}}) \biggr\}\biggr].
\end{align}
Note that (\ref{eq23}) is derived by considering imperfect SIC (i.e., 0$<\kappa_2\leq$1). The exact EC of $s_3$ under perfect SIC is derived as follows.

With perfect SIC, $Y\triangleq\gamma^{t_2}_{s_3}$ can be written as $Y\triangleq\vartheta_3\rho |h_1|^2$. The CDF of $Y$ is therefore given by 
\begin{align}\label{eq24}
F_Y\left(y\right) &=1-e^{-\frac{y}{\vartheta_3\rho\lambda_{1}}}.
\end{align}
Hence, the exact EC associated with $s_3$ is obtained as 
\begin{align}\label{eq25}
\bar{C}^{\textup{ex}}_{3,\textup{p}}
&=\frac{\log_2e}{2}
\displaystyle \int^{\infty}_{0}\frac{1-F_Y\left(y\right)}{\left(1+y\right)}dy
=-\frac{\log_2e}{2}e^{r}\text{Ei}\left(-r\right),
\end{align}
where $r=\frac{1}{\vartheta_3\rho\lambda_{1}}$.
\subsection{Ergodic sum capacity of DU-CNOMA}
Using (\ref{eq14}), (\ref{eq19}), and (\ref{eq23}), the exact closed-form expression of ESC of the proposed DU-CNOMA protocol under imperfect SIC can be written by
\begin{align}\label{eq26}
\bar{C}^{\textup{prop}}_{\textup{sum, ip}}
&=\bar{C}^{\textup{ex}}_{1}+\bar{C}^{\textup{ex}}_{2}+\bar{C}^{\textup{ex}}_{3}.
\end{align}
Conversely, using (\ref{eq14}), (\ref{eq21}), and (\ref{eq25}), the exact closed-form expression of ESC of the proposed DU-CNOMA protocol under perfect SIC can be written by
\begin{align}\label{eq27}
\bar{C}^{\textup{prop}}_{\textup{sum, p}}
&=\bar{C}^{\textup{ex}}_{1}+\bar{C}^{\textup{ex}}_{2,p}+\bar{C}^{\textup{ex}}_{3,p}.
\end{align}
%
%For negligible values of a variable $k$, $\text{Ei}\left(-k\right)\approx E_c+\text{ln}\left(k\right)$ and $e^k\approx1+k$ can be assumed wherein $E_c$ is the Euler constant. By applying these approximations to (\ref{eq20}), the asymptotic expression of CRDU-CNOMA can be derived as 
%\begin{align}\label{eq21a}
%\bar{C}^{prop}_{sum}
%&\simeq-\frac{\log_2e}{2}\left[\left\{1+\frac{\chi}{\rho}\right\}\left\{E_c+\text{ln}\left(\frac{\chi}{\rho}\right)\right\}+\left(1+r\right)\left\{E_c+\text{ln}\left(r\right)\right\}\right]\nonumber\\
%&+\frac{\log_2e}{2}\left\{1+\frac{\chi}{\phi_2\rho}\right\}\left\{E_c+\text{ln}\left(\frac{\chi}{\phi_2\rho}\right)\right\}+\frac{g\log_2e}{2\left(g-1\right)}\nonumber\\
%&\left[-\left(1+q\right)\left\{E_c+\text{ln}\left(q\right)\right\}+\left(1+gq\right)\left\{E_c+\text{ln}\left(gq\right)\right\}\right]
%\end{align} 
%where $\chi=\left(\frac{1}{\lambda_{1}}+\frac{1}{\lambda_{2}}\right)$ is considered for the simplicity.
\section{Outage probability analysis}
 According to the required quality of service, $C_{t_1}$, $C_{t_2}$, and $C_{t_3}$ are assumed to be the predetermined target date rate thresholds of the symbols $s_1$, $s_2$, and $s_3$, respectively. The closed-form expressions of outage probabilities related to  $s_1$, $s_2$, and $s_3$ are provided over independent Rayleigh fading channel in the following subsections.
\subsection{Outage probability of symbol $s_1$}
The OP of symbol $s_1$ is given by
\begin{align}\label{eq28}
P_{\textup{out}, s_1} &=P_r\{\gamma^{t_1}_{s_1}<R_{t_1}\}\nonumber\\
&=1-e^{-\frac{R_{t_1}}{\lambda_{1}\rho(\phi_1-\phi_2R_{t_1})}},
\end{align}
where $R_{t_1}=2^{2C_{t1}}-1$ and $\frac{R_{t_1}}{R_{t_1}+1}<\phi_1<1$.
\subsection{Outage probability of symbol $s_2$}
The OP of symbol $s_2$ is given by
\begin{align}\label{eq29}
P_{\textup{out}, s_2} &=1-P_r\{\min(\gamma^{t_1}_{s_2}, \gamma^{t_2}_{s_2})>R_{t_2}\}P_r\{\gamma^{t_1}_{s_1\rightarrow s_2}>R_{t_1}\}\nonumber\\
&=1-\frac{\phi_2\lambda_{2}\vartheta_2\lambda_{3}}{(\phi_2\lambda_{2}+\phi_1\kappa_1\lambda_{2}R_{t_2})(\vartheta_2\lambda_{3}+\vartheta_3\lambda_{1}R_{t_2})}\nonumber \\
&\times e^{-\frac{R_{t_2}}{\phi_2\rho\lambda_{2}}-\frac{R_{t_2}}{\vartheta_2\rho\lambda_{3}}-\frac{R_{t_1}}{\lambda_{2}\rho(\phi_1-\phi_2R_{t_1})}},
\end{align}
where $R_{t_2}=2^{2C_{t2}}-1$. Now, by putting $\kappa_1$=0, the OP of $s_2$ under perfect SIC can be expressed as 
\begin{align}\label{eq30}
P^{\textup{p}}_{\textup{out}, s_2} &=1-\frac{\vartheta_2\lambda_{3}}{(\vartheta_2\lambda_{3}+\vartheta_3\lambda_{1}R_{t_2})}e^{-\frac{R_{t_2}}{\phi_2\rho\lambda_{2}}-\frac{R_{t_2}}{\vartheta_2\rho\lambda_{3}}-\frac{R_{t_1}}{\lambda_{2}\rho(\phi_1-\phi_2R_{t_1})}}
\end{align} 
\subsection{Outage probability of symbol $s_3$}
The OP of symbol $s_3$ is given by
\begin{align}\label{eq31}
P_{\textup{out}, s_3}&=1-\frac{\vartheta_3 \lambda_{1}}{\vartheta_3 \lambda_{1}+\vartheta_2\kappa_2\lambda_{3}R_{t_3}}e^{-\frac{R_{t_3}}{\vartheta_3\rho\lambda_{1}}},
\end{align}
where $R_{t_3}=2^{2C_{t3}}-1$. Now, by putting $\kappa_2$=0, the OP of $s_3$ under perfect SIC can be expressed as 
\begin{align}\label{eq32}
P^{\textup{p}}_{\textup{out}, s_3} &=1-e^{-\frac{R_{t_3}}{\vartheta_3\rho\lambda_{1}}}.
\end{align}
\subsection{Diversity Order Computation}
To investigate more insight into the system
outage performance, this section demonstrates DO for each symbol in the proposed DU-NOMA. By using $\lim\limits_{\rho\rightarrow\infty}-\frac{\log P_{out}}{\log\rho}$~\cite{28} in the high SNR regime, DOs related to each of the symbols can be computed according to the  following Lemma. 
\begin{lemma}
	Consider $\rho\rightarrow\infty$ and $e^{-x}=1-x$ in the high SNR regime. Hence, DOs of $s_1$, $s_2$, and $s_3$ are respectively expressed as
	\begin{align}\label{eqDo}
	D_{s_1}=&1,\nonumber\\
	D_{s_2} =& 
	\begin{cases} 
      0, \textup{imperfect SIC}\\
	  0, \textup{perfect SIC}
	\end{cases}
	\nonumber\\
	D_{s_3} =& 
	\begin{cases} 
	0, \textup{imperfect SIC}\\
	1, \textup{perfect SIC}
	\end{cases}
	\end{align}
\end{lemma}
\begin{proof}
%\subsubsection{Diversity Order related to symbol $s_1$}
In high SNR, the OP related to symbol $s_1$ can be approximated as 
\begin{align}\label{eq33}
P^\infty_{\textup{out}, s_1} &=\frac{R_{t_1}}{\lambda_{1}\rho(\phi_1-\phi_2R_{t_1})}\approx\frac{1}{\rho}.
\end{align} 
So, the DO related to symbol $s_1$ is derived as 
\begin{align}\label{eq34}
D_{s_1} &=\lim\limits_{\rho\rightarrow\infty}-\frac{\log 1-\log\rho}{\log\rho}=1.
\end{align} 
%\subsubsection{Diversity Order related to symbol $s_2$}
Moreover, the OP related to symbol $s_2$ for large $\rho$ under imperfect and perfect SIC can be respectively approximated as
\begin{align}
P^\infty_{\textup{out}, s_2}&=1-\frac{\phi_2\lambda_{2}\vartheta_2\lambda_{3}}{(\phi_2\lambda_{2}+\phi_1\kappa_1\lambda_{2}R_{t_2})(\vartheta_2\lambda_{3}+\vartheta_3\lambda_{1}R_{t_2})},\label{eq35}\\
P^{\textup{p},\infty}_{\textup{out}, s_2} &=1-\frac{\vartheta_2\lambda_{3}}{(\vartheta_2\lambda_{3}+\vartheta_3\lambda_{1}R_{t_2})}.\label{eq36}
\end{align}
Imperfect SIC and inter-symbol interference effects cause the OP related to $s_3$ settling in the high SNR that creates error floor. Using (\ref{eq35}) and (\ref{eq36}), the DO related to symbol $s_2$ under imperfect and perfect SIC respectively can be found as 
\begin{align}
D_{s_2} &=\lim\limits_{\rho\rightarrow\infty}-\frac{\log P^\infty_{\textup{out}, s_2}}{\log\rho}=0,\label{eq37}\\
D^\textup{p}_{s_2} &=\lim\limits_{\rho\rightarrow\infty}-\frac{\log P^{\textup{p},\infty}_{\textup{out}, s_2}}{\log\rho}=0.\label{eq38}
\end{align} 
%\subsubsection{Diversity Order related to symbol $s_3$}
Lastly, in high SNR, the OP related to symbol $s_3$ under imperfect and perfect SIC can be respectively approximated as
\begin{align}
P^\infty_{\textup{out}, s_3}&=1-\frac{\vartheta_3 \lambda_{1}}{\vartheta_3 \lambda_{1}+\vartheta_2\kappa_2\lambda_{3}R_{t_3}},\label{eq39}\\
P^{\textup{p},\infty}_{\textup{out}, s_3} &=\frac{R_{t_3}}{\vartheta_3\rho\lambda_{1}}\approx\frac{1}{\rho}.\label{eq40}
\end{align}
Therefore, the DO related to $s_3$ under imperfect and perfect SIC can be respectively computed by using (\ref{eq39}) and (\ref{eq40}) as
\begin{align}
D_{s_3} &=\lim\limits_{\rho\rightarrow\infty}-\frac{\log P^\infty_{\textup{out}, s_3}}{\log\rho}=0,\label{eq41}\\
D^\textup{p}_{s_3} &=\lim\limits_{\rho\rightarrow\infty}-\frac{\log P^{\textup{p},\infty}_{\textup{out}, s_3}}{\log\rho}=\lim\limits_{\rho\rightarrow\infty}-\frac{\log 1-\log\rho}{\log\rho}=1.\label{eq42}
\end{align}
The DO related to $s_3$ under realistic imperfect SIC assumption becomes zero whereas it becomes one under perfect SIC assumption.  
\end{proof} 
\section{Outage capacity analysis}
This section presents analytical derivation of OC for each symbol in the proposed DU-CNOMA over independent Rayleigh fading channels. The OC is defined	as the data rate that can be attained if a system is allowed to be in outage with probability $\Upsilon$. For wireless environments with deep fading situation, it is a critical performance metric. The OC of each symbol can be derived from their corresponding OP by using \cite[eq. (2.68)]{29} as shown in the following subsections.
\subsection{Outage capacity of $s_1$}
The OC $C_{t_1}$ related to $s_1$, with specified OP $\Upsilon_1$ can be computed from (\ref{eq28}) as
\begin{align}\label{eq43}
\Upsilon_1=1-e^{-\frac{R_{t_1}}{\lambda_{1}\rho(\phi_1-\phi_2R_{t_1})}}\nonumber\\
e^{-\frac{R_{t_1}}{\lambda_{1}\rho(\phi_1-\phi_2R_{t_1})}}=1-\Upsilon_1 \nonumber\\
-\frac{R_{t_1}}{\lambda_{1}\rho(\phi_1-\phi_2R_{t_1})}=\ln (1-\Upsilon_1) \nonumber\\
\{(\lambda_{1}\phi_2\rho \ln(1-\Upsilon_1)-1\}R_{t_1}=\lambda_{1}\phi_1\rho
 \ln(1-\Upsilon_1)\nonumber\\
 2^{2C_{t1}}-1=\frac{\lambda_{1}\phi_1\rho
 	\ln(1-\Upsilon_1)}{(\lambda_{1}\phi_2\rho \ln(1-\Upsilon_1)-1}\nonumber\\
 C_{t1}=\frac{1}{2}\log_2\biggl(1+\frac{\lambda_{1}\phi_1\rho
 	\ln(1-\Upsilon_1)}{(\lambda_{1}\phi_2\rho \ln(1-\Upsilon_1)-1}\biggr).
\end{align}
\subsection{Outage capacity of $s_2$}
Using $e^x\approx 1+x$ at high $\rho$, the OC $C_{t_2}$ related to $s_2$, with specified OP $\Upsilon_2$ can be computed from (\ref{eq29}) as
\begin{align}\label{eq44}
\Upsilon_2 &=1-\frac{\phi_2\lambda_{2}\vartheta_2\lambda_{3}}{(\phi_2\lambda_{2}+\phi_1\kappa_1\lambda_{2}R_{t_2})(\vartheta_2\lambda_{3}+\vartheta_3\lambda_{1}R_{t_2})}\nonumber \\
&\times \biggl(1-\frac{R_{t_2}}{\phi_2\rho\lambda_{2}}-\frac{R_{t_2}}{\vartheta_2\rho\lambda_{3}}\biggr)\nonumber\\
\Upsilon_2 &=1-\frac{GH}{(G+IR_{t_2})(H+JR_{t_2})}\biggl(1-\frac{R_{t_2}}{G\rho}-\frac{R_{t_2}}{H\rho}\biggr),
\end{align}
conditioned on $\gamma^{t_1}_{s_1\rightarrow s_2}>R_{t_1}$, where $G=\phi_2\lambda_{2}$, $H=\vartheta_2\lambda_{3}$, $I=\phi_1\kappa_1\lambda_{2}$, and $J=\vartheta_3\lambda_{1}$. After some algebraic simplifications, (\ref{eq44}) can be rewritten as
\begin{align}\label{eq45}
IJ\rho\left(1-\Upsilon_2\right)R_{t_2}^2+\left\{\left(GJ+HI\right)\left(1-\Upsilon_2\right)\rho+H+G\right\}R_{t_2}
+\left(-GH\rho\Upsilon_2\right)=&0\nonumber\\
KR_{t_2}^2+LR_{t_2}+M=&0
\end{align}
where $K=IJ\rho\left(1-\Upsilon_2\right)$, $L=\left(GJ+HI\right)\left(1-\Upsilon_2\right)\rho+H+G$, $M=-GH\rho\Upsilon_2$ are assumed. Solving (\ref{eq45}) and considering feasible root, $C_{t_2}$ can be obtained as 
\begin{align}\label{eq46}
R_{t_2}&=\frac{-L+\sqrt{L^2-4KM}}{2K}\nonumber\\
2^{2C_{t_2}}-1&=\frac{-L+\sqrt{L^2-4KM}}{2K}\nonumber\\
C_{t_2}&=\frac{1}{2}\text{log}_2\left(1+\frac{-L+\sqrt{L^2-4KM}}{2K}\right)
\end{align}

On the other hand, the OC $C_{t_2}$ under perfect SIC can be computed from (\ref{eq30}) as
%\begin{align}\label{eq37aa}
%\Upsilon_2 &=1-\frac{\vartheta_2\lambda_{3}}{(\vartheta_2\lambda_{3}+\vartheta_3\lambda_{1}R_{t_2})}{\biggl(1-\frac{R_{t_2}}{\phi_2\rho\lambda_{2}}-\frac{R_{t_2}}{\vartheta_2\rho\lambda_{3}}\biggr)}
%\end{align} 
\begin{align}\label{eq47}
\Upsilon_2 &=1-\frac{\vartheta_2\lambda_{3}}{(\vartheta_2\lambda_{3}+\vartheta_3\lambda_{1}R_{t_2})}{\biggl(1-\frac{R_{t_2}}{\phi_2\rho\lambda_{2}}-\frac{R_{t_2}}{\vartheta_2\rho\lambda_{3}}\biggr)}\nonumber\\
&=1-\frac{H}{(H+JR_{t_2})}\biggl(1-\frac{R_{t_2}}{G\rho}-\frac{R_{t_2}}{H\rho}\biggr)\nonumber\\
R_{t_2}&=\frac{GH\rho\Upsilon_2}{GJ\rho\left(1-\Upsilon_2\right)+G+H}\nonumber\\
C_{t_2}&=\frac{1}{2}\text{log}_2\left(1+\frac{GH\rho\Upsilon_2}{GJ\rho\left(1-\Upsilon_2\right)+G+H}\right)
\end{align} 
\subsection{Outage capacity of $s_3$}
Using $e^x\approx 1+x$ at high $\rho$, the OC $C_{t_3}$ related to $s_3$, with specified OP $\Upsilon_3$ can be computed from (\ref{eq31}) as
\begin{align}\label{eq48}
\Upsilon_3&=1-\frac{\vartheta_3 \lambda_{1}}{\vartheta_3 \lambda_{1}+\vartheta_2\kappa_2\lambda_{3}R_{t_3}}{(1-\frac{R_{t_3}}{\vartheta_3\rho\lambda_{1}})}\nonumber\\
\Upsilon_3&=1-\frac{\vartheta_3 \lambda_{1}(\vartheta_3\rho\lambda_{1}-R_{t_3})}{(\vartheta_3 \lambda_{1}+\vartheta_2\kappa_2\lambda_{3}R_{t_3})\vartheta_3\rho\lambda_{1}}\nonumber\\
\Upsilon_3&=\frac{\vartheta_2\kappa_2\lambda_{3}R_{t_3}\rho+R_{t_3}}{(\vartheta_3 \lambda_{1}+\vartheta_2\kappa_2\lambda_{3}R_{t_3})\rho}\nonumber\\
R_{t_3}&=\frac{\vartheta_3 \lambda_{1}\rho \Upsilon_3}{1+\vartheta_2\kappa_2\lambda_{3}\rho-\vartheta_2\kappa_2\lambda_{3}\rho\Upsilon_3}\nonumber\\
C_{t_3}&=\frac{1}{2}\log_2\biggl(1+\frac{\vartheta_3 \lambda_{1}\rho \Upsilon_3}{1+\vartheta_2\kappa_2\lambda_{3}\rho(1-\Upsilon_3)}\biggr)
\end{align}
On the other hand, the OC $C_{t_3}$ under perfect SIC can be computed from (\ref{eq32}) as
\begin{align}\label{eq49}
\Upsilon_3 &=1-e^{-\frac{R_{t_3}}{\vartheta_3\rho\lambda_{1}}}\nonumber \\
e^{-\frac{R_{t_3}}{\vartheta_3\rho\lambda_{1}}}&=1-\Upsilon_3\nonumber \\
R_{t_3}&=-\vartheta_3\rho\lambda_{1}\ln (1-\Upsilon_3)\nonumber\\
2^{2C_{t3}}&=1-\vartheta_3\rho\lambda_{1}\ln (1-\Upsilon_3)\nonumber\\
C_{t3}&=\frac{1}{2}\log_2(1-\vartheta_3\rho\lambda_{1}\ln (1-\Upsilon_3))
\end{align}
\subsection{Outage sum capacity}
Using (\ref{eq43}), (\ref{eq46}), and (\ref{eq48}), the OSC of the proposed DU-CNOMA under imperfect SIC is given by 
\begin{align}\label{eq50}
C^{\textup{out}}_{\textup{sum, ip}}&= (\ref{eq43})+ (\ref{eq46})+ (\ref{eq48}).
\end{align}
Conversely, Using (\ref{eq43}), (\ref{eq47}), and (\ref{eq49}), the OSC of the proposed DU-CNOMA under perfect SIC is given by 
\begin{align}\label{eq51}
C^{\textup{out}}_{\textup{sum, p}}&= (\ref{eq43})+ (\ref{eq47})+ (\ref{eq49}).
\end{align}
\begin{figure}[t]
	\centering 
	\includegraphics[width=3.5in,height=3.5in,keepaspectratio]{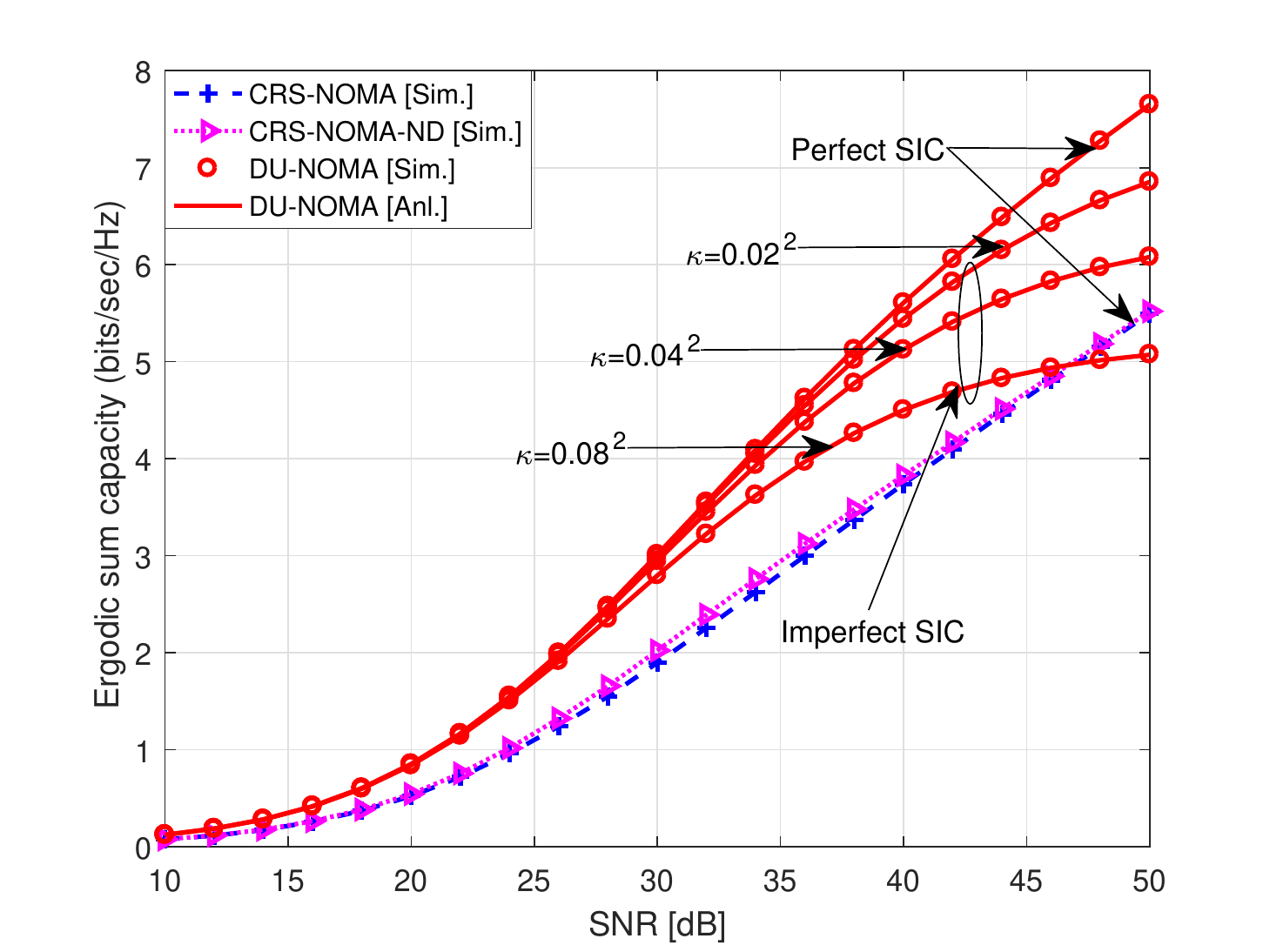}
	\caption{ESC comparison between CRS-NOMA~\cite{19}, CRS-NOMA-ND~\cite{21}, and proposed DU-CNOMA w.r.t SNR $\rho$. $\phi_1$=0.9, $\phi_2$=0.1, $\vartheta_2$=1, and $\vartheta_3$=0.7.}
	\label{Fig2}
\end{figure}
\begin{figure}[H]
	\centering
	\includegraphics[width=3.5in,height=3.5in,keepaspectratio]{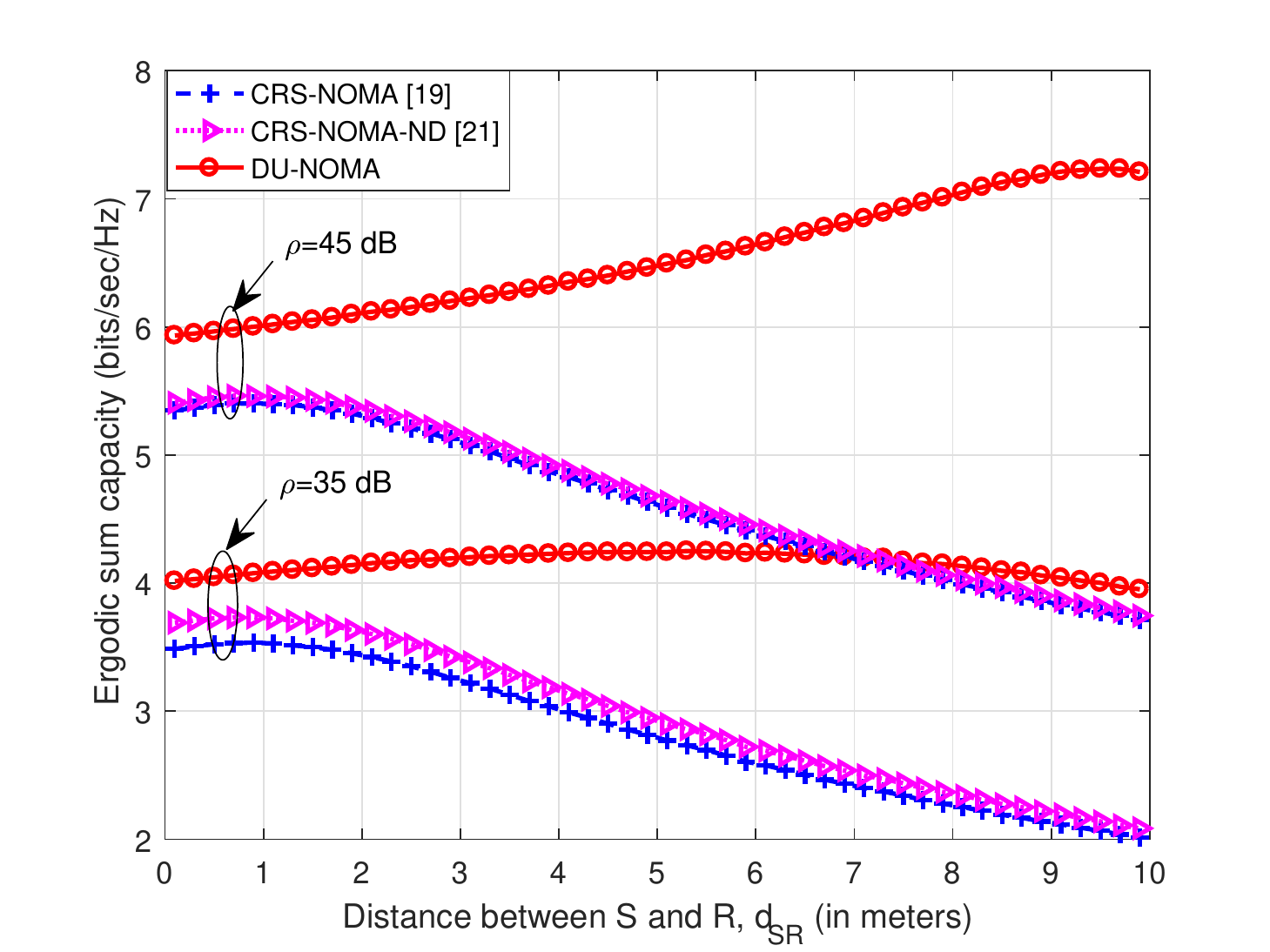}
	\caption{ESC comparison between CRS-NOMA~\cite{19}, CRS-NOMA-ND~\cite{21}, and proposed DU-CNOMA w.r.t distance (in meters) from S to R, $d_{SR}$ under perfect SIC. $\phi_1$=0.9, $\phi_2$=0.1, $\vartheta_2$=1, and $\vartheta_3$=0.7.}
	\label{Fig3}
\end{figure}
\section{Energy Efficiency}
To study the performance of future wireless networks, EE can be another important performance metric. Hence, in this section, we demonstrate EE of the proposed DU-NOMA protocol. The expression of EE $\eta$ can be written by
\begin{align}\label{ee1}
\eta=\frac{\textup{Total data rate of the NOMA system}}{\textup{Total energy consumption}}.
\end{align}
Now, using (\ref{eq26}) and (\ref{eq27}), EE of DU-NOMA considering ESC or delay-tolerant transmission mode under imperfect SIC and perfect SIC can be respectively given by~\cite{26}
\begin{align}\label{ee2}
\eta^{\textup{ESC}}_{\textup{DU-NOMA, ip}}=&\frac{\bar{C}^{\textup{prop}}_{\textup{sum, ip}}}{\frac{T}{2}P_s+\frac{T}{2}P_r}=\frac{2\bar{C}^{\textup{prop}}_{\textup{sum, ip}}}{TP_s+TP_r},\nonumber\\
\eta^{\textup{ESC}}_{\textup{DU-NOMA, p}}=&\frac{\bar{C}^{\textup{prop}}_{\textup{sum, p}}}{\frac{T}{2}P_s+\frac{T}{2}P_r}=\frac{2\bar{C}^{\textup{prop}}_{\textup{sum, p}}}{TP_s+TP_r}.
\end{align}
where $T$ represents the time of a complete transmission. Similarly, using (\ref{eq50}) and (\ref{eq51}), EE of DU-NOMA considering OSC under imperfect SIC and perfect SIC can be respectively given by
\begin{align}\label{ee3}
\eta^{\textup{OSC}}_{\textup{DU-NOMA, ip}}=&\frac{2C^{\textup{out}}_{\textup{sum, ip}}}{TP_s+TP_r},\nonumber\\
\eta^{\textup{OSC}}_{\textup{DU-NOMA, p}}=&\frac{2C^{\textup{out}}_{\textup{sum, p}}}{TP_s+TP_r}.
\end{align}
Moreover, for the purpose of comparison, EE of CRS-NOMA and CRS-NOMA-ND under perfect SIC can also be investigated using \cite[eq. (10)]{19} and \cite[eq. (20)]{21} into (56), respectively considering delay-tolerant transmission mode.
\begin{figure}[t]
	\centering
	\includegraphics[width=3.5in,height=3.5in,keepaspectratio]{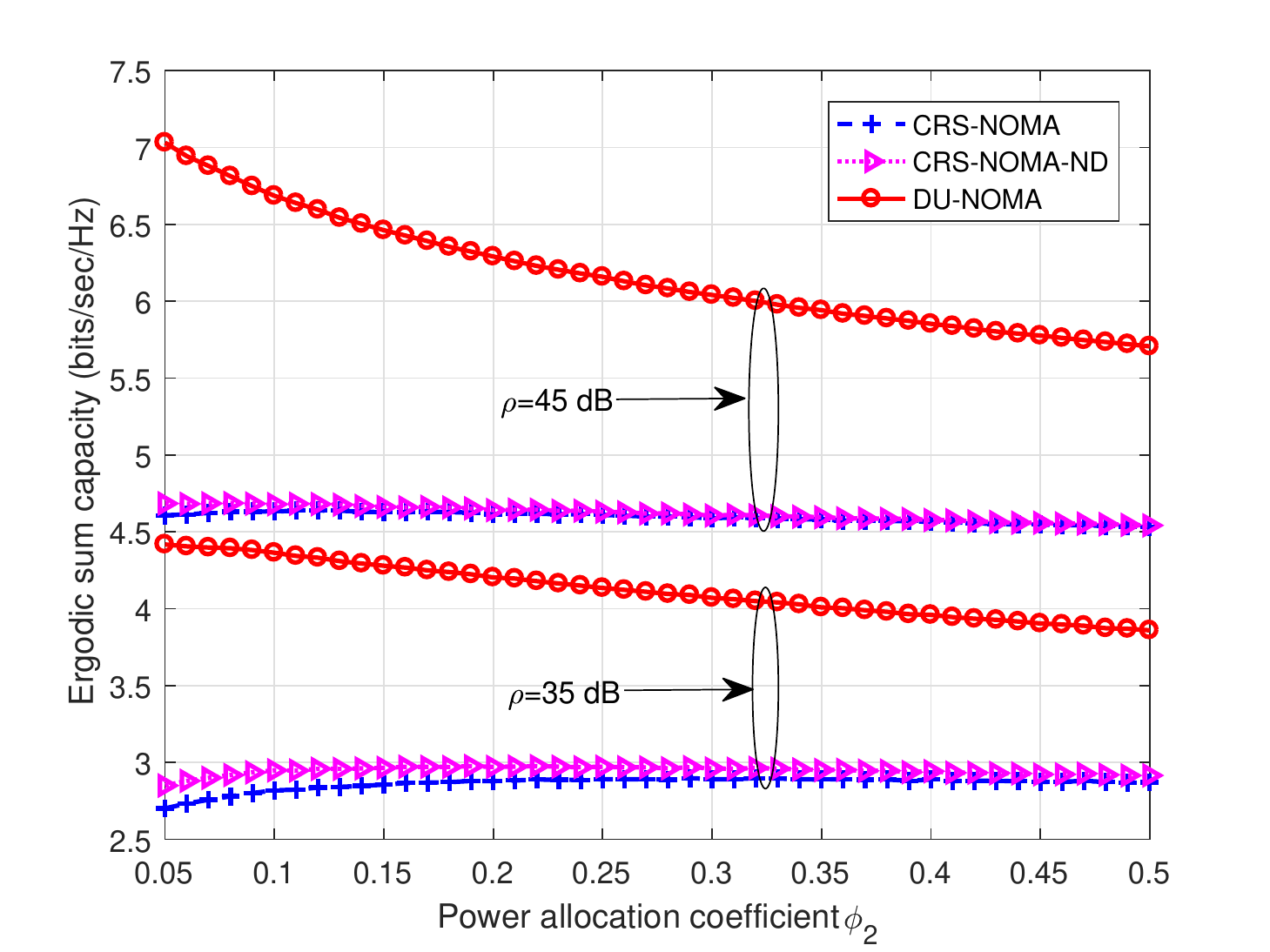}
	\caption{ESC comparison between CRS-NOMA~\cite{19}, CRS-NOMA-ND~\cite{21}, and proposed DU-CNOMA w.r.t power allocation coefficient $\phi_2$ under perfect SIC. $\vartheta_2$=1 and $\vartheta_3$=0.7.}
	\label{Fig4}
\end{figure}
\section{Numerical results and discussions}
This section presents simulation (Sim.) and analytical (Anl.) results of our proposed DU-CNOMA protocol. In each case, analytical result matches well with simulation result and it confirms the correctness of the author's analysis presented here. For comparison purpose, the simulation results for CRS-NOMA~\cite{19} and CRS-NOMA-ND~\cite{21} are also presented. It should be mentioned that analytical derivations for OP and OSC are not provided in~{\cite{19, 21}}. Throughout the simulation, it is assumed that $\nu$=2, $d_{SD}$=10 meters, $d_{SR}=d_{SD}/2$, $d_{RD}=1-d_{SR}$, $\phi_1$=0.9, $\phi_2$=0.1, $\Upsilon_1=\Upsilon_2=\Upsilon_3=\Upsilon$, $\vartheta_2$=1, $\vartheta_3$=0.7, and $\kappa_1=\kappa_2=\kappa$, unless otherwise specified. Note that fixed power allocation method as in~\cite{19, 20, 21} is assumed for the proposed protocol.

It should be mentioned that the distance used in the simulations is arbitrary. Although, we have assumed that the distance between S and D is 10 m, the proposed protocol is viable for any distance between S and D by tuning different parameters such as transmit power, bandwidth, path loss exponent etc. Furthermore, to meet the demanded traffic in 5G or future wireless systems, network densification by the deployment of ultra-dense small cells is one of the most effective techniques. Therefore, small cells with inter-site distance of 5, 10, 20 m etc. are expected to be common in future wireless systems~\cite{30, 31}.
\begin{figure}[t]
	\centering
	\includegraphics[width=3.5in,height=3.5in,keepaspectratio]{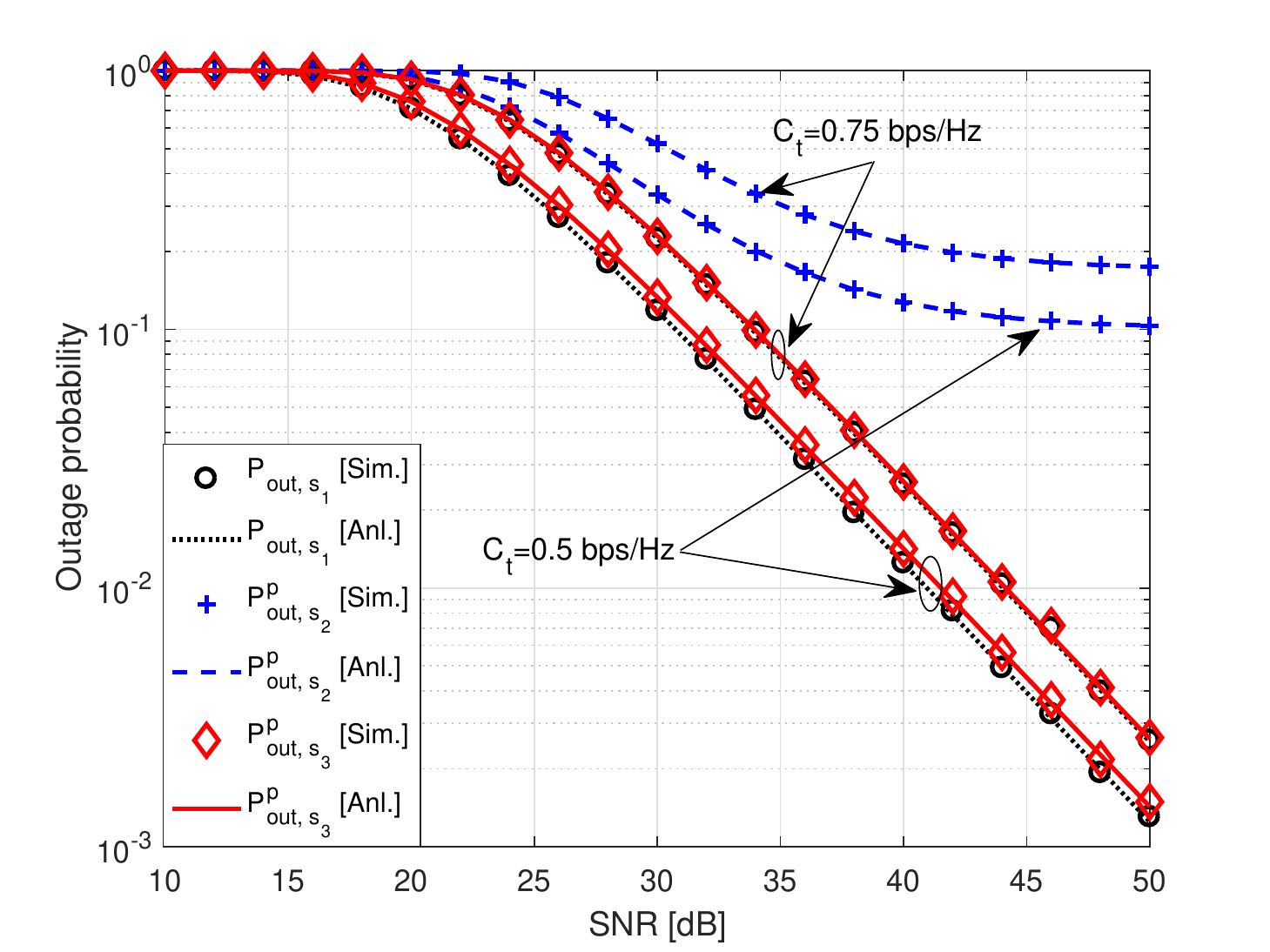}
	\caption{OP of the proposed DU-CNOMA w.r.t SNR $\rho$ under perfect SIC. $\phi_1$=0.9, $\phi_2$=0.1, $\vartheta_2$=1, and $\vartheta_3$=0.7.}
	\label{Fig5}
\end{figure}
\begin{figure}[H]
	\centering
	\includegraphics[width=3.5in,height=3.5in,keepaspectratio]{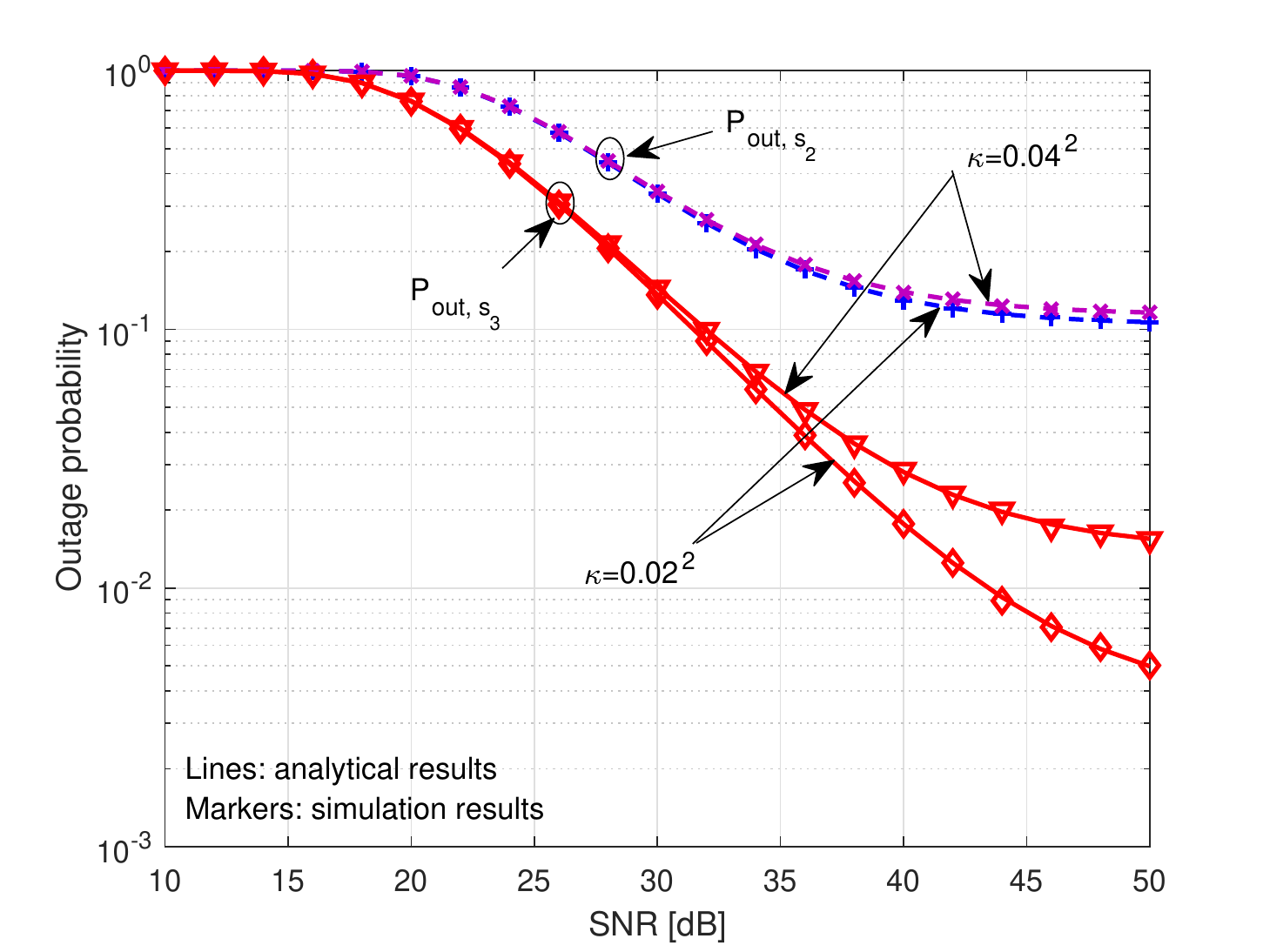}
	\caption{OP of the proposed DU-CNOMA w.r.t SNR $\rho$ under imperfect SIC. $\phi_1$=0.9, $\phi_2$=0.1, $\vartheta_2$=1, $\vartheta_3$=0.7, and $C_t$=0.5 bps/Hz .}
	\label{Fig6}
\end{figure}
\subsection{Ergodic capacity}
ESC versus SNR behavior of DU-CNOMA, CRS-NOMA, and CRS-NOMA-ND is shown in Fig. \ref{Fig2}. Performance of the proposed DU-CNOMA is executed under two conditions, i.e., perfect SIC and imperfect SIC. Note that only perfect SIC is considered in CRS-NOMA and CRS-NOMA-ND. For the case of perfect SIC, it is observed from the figure that DU-CNOMA significantly outperforms all other existing protocols. However, with the increasing amount of residual interference the performance of DU-CNOMA starts degrading which causes it to exhibit a saturated value at high $\rho$ values. For example, the performance of the proposed protocol becomes worse for $\kappa = 0.04^2$ than $\kappa = 0.02^2$. Therefore, at high $\rho$, the adverse impact of residual interference on DU-CNOMA causes it to achieve less ESC than existing methods. Therefore, it is suggested that an efficient interference cancellation technique can significantly improve the performance of DU-NOMA, particularly at medium to high $\rho$. Lastly, strong agreement between simulation and analytical results verifies the appropriateness of the ESC analysis. 
\par
ESC behavior for varying relay position between source and destination, $d_{\textup{SR}}$ (in meters) is demonstrated in Fig. \ref{Fig3}, under perfect SIC. ESC versus $d_{\textup{SR}}$ performance of DU-CNOMA is compared with CRS-NOMA and CRS-NOMA-ND protocols for two different $\rho$ values, i.e., $\rho$ = 35 and 45 dBs. For both cases, proposed protocol achieves better ESC than existing protocols irrespective of the relay position. In addition, ESC of DU-CNOMA becomes far better than others for the increasing distance between source and relay. However, the ESC of DU-CNOMA becomes maximum at a certain $d_{\textup{SR}}$. For example, this behavior is bounded by a maximum $d_{\textup{SR}}$ value (e.g., around $d_{\textup{SR}}$ = 5 for $\rho$ = 35 dB and around $d_{\textup{SR}}$ = 9.5 for $\rho$ = 45 dB).
\par 
ESC with respect to (w.r.t) the power allocation coefficient $\phi_2$ is shown in Fig. \ref{Fig4} for two cases of $\rho$, where $\rho$ = 35 dB and $\rho$ = 45 dB. It is demonstrated that ESC performance of DU-NOMA degrades with the increase of $\phi_2$, whereas $\phi_2$ has a slight impact of the performance of CRS-NOMA and CRS-NOMA-ND. Further, ESC of the proposed DU-CNOMA protocol is higher than existing protocols for all feasible values of $\phi_2$. It is also clear from the figure that ESC of the proposed protocol is higher for $\rho$ = 45 dB than 35 dB.
\begin{figure}[t]
	\centering
	\includegraphics[width=3.5in,height=3.5in,keepaspectratio]{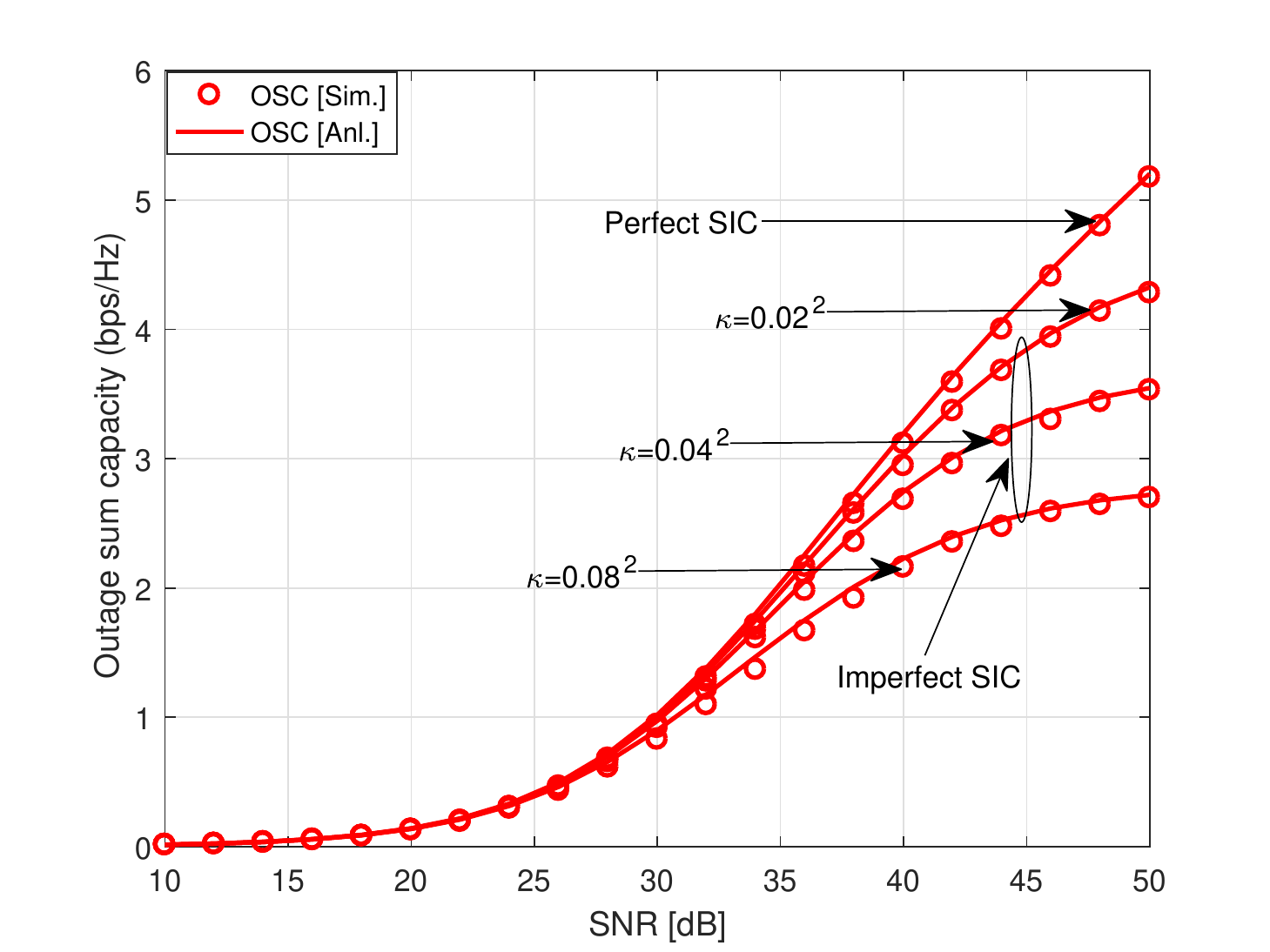}
	\caption{OSC of the proposed DU-CNOMA w.r.t SNR $\rho$. $\phi_1$=0.9, $\phi_2$=0.1, $\vartheta_2$=1, $\vartheta_3$=0.7, and $\Upsilon$=0.1.}
	\label{Fig7}
\end{figure}
\begin{figure}[H]
	\centering
	\includegraphics[width=3.5in,height=3.5in,keepaspectratio]{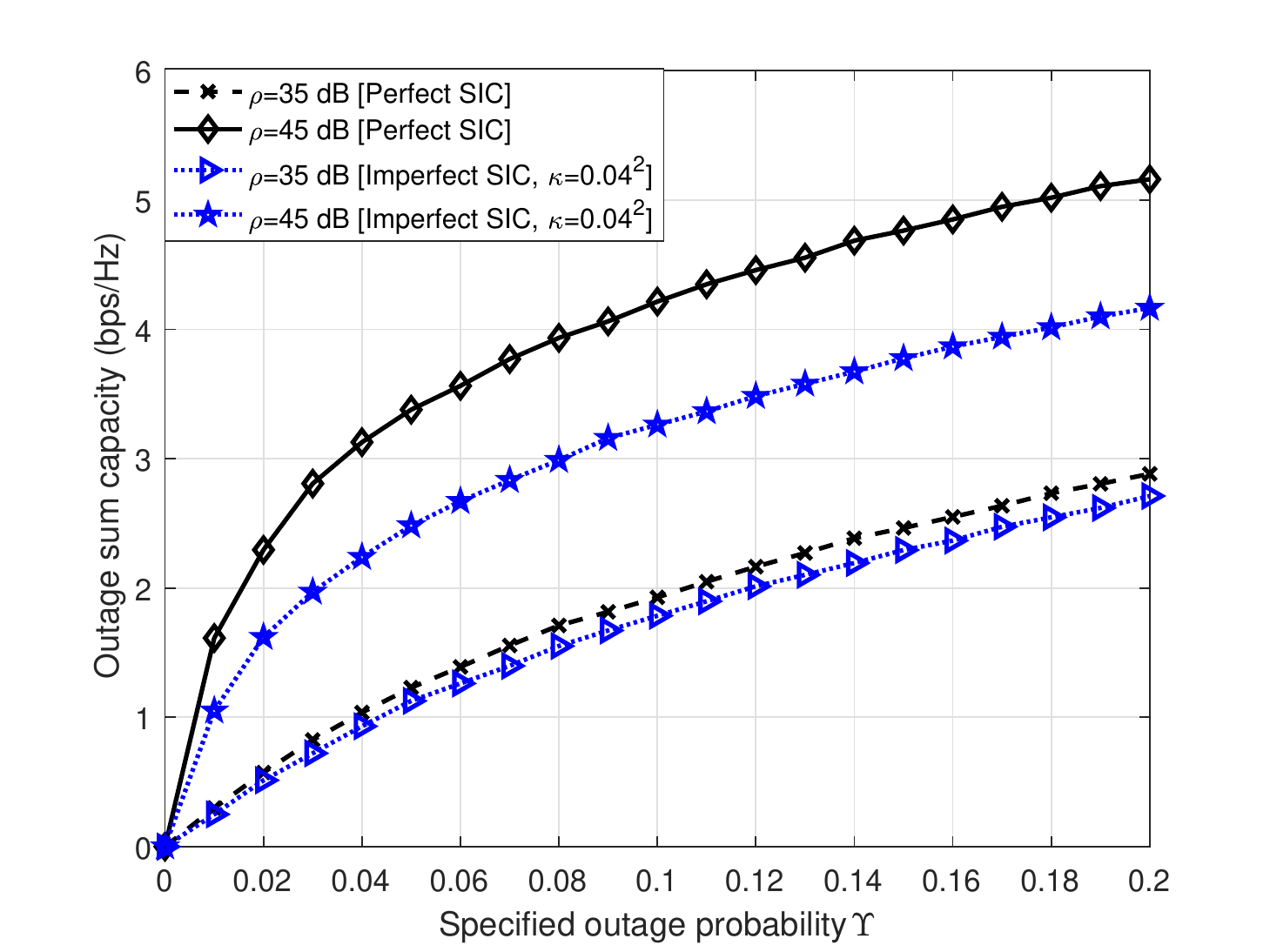}
	\caption{OSC of the proposed DU-CNOMA w.r.t specified outage probability $\Upsilon$ under perfect SIC. $\phi_1$=0.9, $\phi_2$=0.1, $\vartheta_2$=1, and $\vartheta_3$=0.7.}
	\label{Fig8}
\end{figure}    
\subsection{Outage probability}
Let $C_{t_1}=C_{t_2}=C_{t_3}=C_t$. OP versus SNR performance of the proposed protocol is demonstrated in Fig. \ref{Fig5} for two different threshold values of target data rate, i.e., $C_\textup{t}$ = 0.5 bps/Hz and 0.75 bps/Hz. Perfect SIC is considered for analyzing all analytical and simulation results. Coincidence of analytical and simulation results for each case refers to the accuracy of OP analysis. It is pointed out that OP becomes better with the increase of SNR. The OPs related to $s_1$ and $s_3$ show better than $s_2$ for a specific $C_t$ when $\rho$ ranges from medium to high. Though OPs related to symbols $s_1$ and $s_3$ decrease linearly with the increase of $\rho$ after a certain $\rho$, OP related to symbol $s_2$ tends to be saturated for high $\rho$ range. The OPs related to $s_1$, $s_2$, and $s_3$ for $C_\textup{t}$ = 0.75 bps/Hz is higher than $C_\textup{t}$ = 0.5 bps/Hz, as expected.
%The reason behind exhibiting this performance by $s_2$ is the %interference effect from other symbols on it.
\par 
By considering imperfect SIC and target data rate $C_\textup{t}$ = 0.5 bps/Hz, OP of DU-CNOMA protocol w.r.t SNR $\rho$ is depicted in Fig. \ref{Fig6}. Only OP versus SNR analysis related to s$_2$ and $s_3$ are compared as the performance related to $s_1$ is not affected by imperfect SIC condition. OP related to $s_3$ is better than $s_2$ at any value of $\rho$ for the considered parameters. OP related to any of the symbols is less for small residual interference (i.e., $\kappa= 0.02^2$) than comparatively large amount of residual interference (i.e., $\kappa = 0.04^2$). Though OP related to $s_3$ shows linear behavior even at high $\rho$ as shown in Fig. \ref{Fig5}, it tends to be saturated at high $\rho$ as shown in Fig. \ref{Fig6} due to the impact of residual interference. 
\begin{figure}[t]
	\centering
	\includegraphics[width=3.5in,height=3.5in,keepaspectratio]{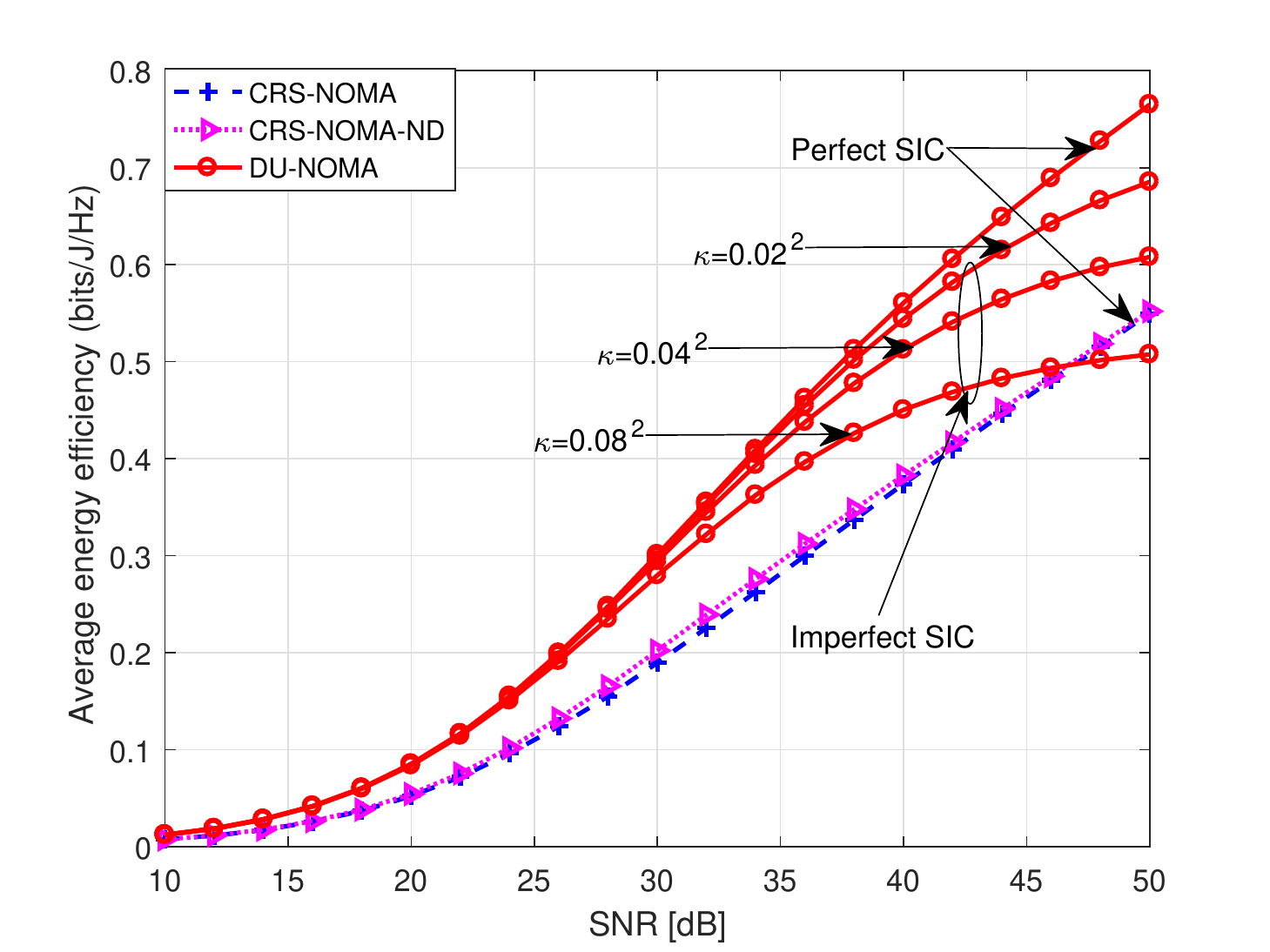}
	\caption{EE comparison between CRS-NOMA~\cite{19}, CRS-NOMA-ND~\cite{21}, and proposed DU-CNOMA w.r.t SNR $\rho$ considering ergodic rate or delay-tolerant transmission. $\phi_1$=0.9, $\phi_2$=0.1, $\vartheta_2$=1, $\vartheta_3$=0.7, T=1 and $P_s=P_r=10$ W.}
	\label{Fig10}
\end{figure}
\subsection{Outage capacity}
10\% OSC of the proposed DU-CNOMA protocol w.r.t SNR $\rho$ is plotted under both perfect and imperfect SIC conditions in Fig. \ref{Fig7}. Three cases of imperfect SIC is considered, i.e., $\kappa= 0.02^2$, $\kappa= 0.04^2$ and $\kappa= 0.08^2$. For perfect SIC condition, OSC of the system linearly increases with the betterment of $\rho$ and maintains it till the high $\rho$. But, for imperfect SIC condition, OSC of the system tends to be saturated at high $\rho$ due to the impact of residual interference. If the effect of residual interference increases, OSC of DU-CNOMA decreases and tends to be saturated at a less value of $\rho$ than for comparatively small residual interference impact.
\par 
OSC behavior w.r.t specified outage probability $\Upsilon$ for the proposed DU-CNOMA protocol is demonstrated in Fig. \ref{Fig8}. Both perfect and imperfect SIC conditions are taken into account and the performance behavior is observed for two different values of $\rho$, i.e., $\rho$ = 35 dB and 45 dB. Fig. \ref{Fig8} depicts that OSC of the system increases with the increase of specified outage probability $\Upsilon$. In addition, OSC goes high for higher $\rho$ (= 45 dB) than lower $\rho$ (= 35 dB). It is also clear that OSC of DU-NOMA shows better for perfect SIC case than imperfect SIC case for a specific value of $\rho$. 
\begin{figure}[H]
	\centering
	\includegraphics[width=3.5in,height=3.5in,keepaspectratio]{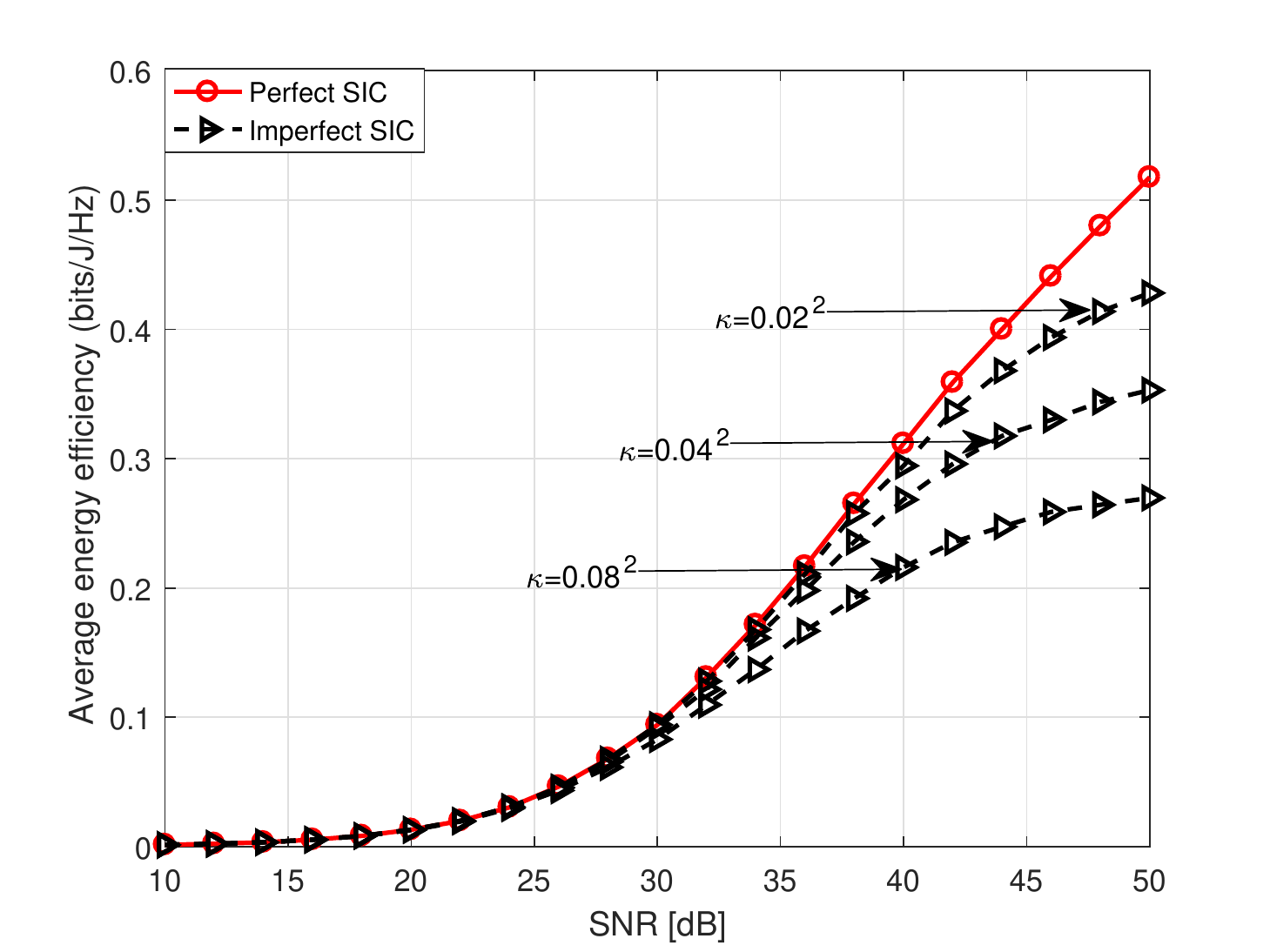}
	\caption{EE of the proposed DU-CNOMA w.r.t SNR $\rho$ considering OSC. $\phi_1$=0.9, $\phi_2$=0.1, $\vartheta_2$=1, $\vartheta_3$=0.7, T=1 and $P_s=P_r=10$ W.}
	\label{Fig11}
\end{figure} 
\subsection{Energy efficiency} 
Fig. 10 shows EE comparison between DU-NOMA, CRS-NOMA, and CRS-NOMA-ND w.r.t. SNR considering delay-tolerant transmission mode or ergodic rate. It is observed that DU-NOMA shows significantly better performance than CRS-NOMA and CRS-NOMA-ND, in terms of EE under perfect SIC. However, under imperfect SIC, the performance gain of DU-NOMA over others depend on the level of residual interference, particularly at high $\rho$. Finally, Fig. 11 shows EE of the proposed DU-NOMA versus SNR considering outage sum rate under both perfect and imperfect SIC. From both figures, it is clear that perfect SIC case outperforms imperfect SIC case, particularly at medium to high $\rho$.
\subsection{Comparative study}
Finally, a comparative study between the proposed DU-NOMA, CRS-NOMA~\cite{19}, and CRS-NOMA-ND~\cite{21} is summarized in Table 1. 
\section{Conclusion and future works}
A cooperative decode-and-forward relaying strategy using the concept of downlink and uplink NOMA has proposed and analyzed in this paper. Under both perfect and imperfect SIC, the performance of the proposed protocol has studied comprehensively, in terms of ESC, OP, EE and OSC over independent Rayleigh fading channels. The closed-form expressions of these system parameters have derived and validated by computer simulation. To get insight into the systems's outage performance, DO for each symbol is also computed. It has shown that the proposed protocol significantly outperforms CRS-NOMA and CRS-NOMA-ND under perfect SIC, whereas under imperfect SIC, performance gains depends on the level of residual interference, particularly at high SNR. Furthermore, hybrid downlink-uplink NOMA for multi-input multi-output systems will be investigated in future works.
\begin{table}[H]
		\caption{A comparative study between the proposed DU-NOMA, CRS-NOMA~\cite{19}, and CRS-NOMA-ND~\cite{21}.}
	\centering 
	\begin{tabular}{ |p{1.8cm} | p{6cm} |  p{6cm} |}
		\hline
		Item & CRS-NOMA~\cite{19} and  CRS-NOMA-ND~\cite{21} & Proposed DU-NOMA \\ \hline
		System Model & A CRS-NOMA consisting of a S, a R and a D is proposed, where direct link between S and D is active.
		\begin{itemize}
			\item A downlink NOMA is exploited in Phase-1 only. In phase-2, only R retransmits the decoded symbol to D.
			\item In contrast to CRS-NOMA~\cite{19}, D uses MRC and another SIC to jointly decode transmitted symbols by S in CRS-NOMA-ND~\cite{21}.
			\item Two symbols are transmitted during two phases.
		\end{itemize}
		 & A different relaying strategy using NOMA (termed as DU-NOMA) is proposed and investigated to further improve the performance of CRS-NOMA system as compared to~\cite{19, 21}.  
		 \begin{itemize}
		 	\item The concept of downlink NOMA is exploited in phase-1, whereas the concept of uplink NOMA is applied in phase-2.
		 	\item Three symbols are transmitted during two phases.
		 \end{itemize}		 	
		 	\\ \hline
		Performance Metrics & 
		\begin{itemize}
			\item ESC
			\item Though simulation results for OP is also shown in~\cite{21},  \cite{21} did not provide any closed-form expressions for OP.
		\end{itemize}
	 & \begin{itemize}
	 	\item ESC 
	 	\item OP 
	 	\item OSC
	 	\item DO and
	 	\item EE
	 \end{itemize} \\ \hline
		Analytical derivations & Analytical derivation for ESC is provided. & Analytical derivations for ESC, OP, OSC, DO, and EE are provided. \\
		\hline
		Numerical results & Simulation results show only ESC in~\cite{19}, whereas ESC and OP in~\cite{21}. & Simulation results show ESC, OP, OSC, and EE.\\
		\hline
		Outcomes & CRS-NOMA~\cite{19} can achieve more spectral efficiency than the conventional CRS~\cite{3} when the SNR is high, and the average channel power of the S-to-R link is better than those of the S-to-D and R-to-D links. On the contrary, \cite{21} outperforms \cite{19}, particularly when the link between S and R is better than the R to D link. & The proposed DU-NOMA is more spectral and energy efficient as compared to CRS-NOMA~\cite{19} and CSR-NOMA-ND~\cite{21}.\\
		\hline
	\end{tabular}
\end{table}
\section*{Disclosure statement}
The author(s) declare(s) no potential conflict of interest regarding the publication of this paper.
\section*{Acknowledgment}
This research was supported by the University of Chittagong, in part by the MSIT (Ministry of Science, ICT), Korea, under the ITRC (Information Technology Research Center) support program (IITP-2018-2014-1-00639) supervised by the IITP (Institute for Information \& communications Technology Promotion). This research was also supported by the Sejong University Research Faculty Program Fund (2019-2021), as well as the Basic Science Research Program through the National Research Foundation of Korea (NRF) funded by the Ministry of Education (2015R1D1A1A01061075).

\bibliography{UpDownNOMA}

\end{document}